\documentclass{lmcs}

\usepackage[utf8x]{inputenc}

\usepackage{hyperref}

\usepackage{xargs}

\usepackage[colorinlistoftodos,color=red!10,prependcaption,textsize=small
]{todonotes}

\usepackage{tikz}
\usetikzlibrary{arrows,automata,intersections,matrix}

\usepackage{amsmath}
\usepackage{amsfonts}
\usepackage{amssymb}
\usepackage{stmaryrd}

\usepackage{multirow}

\newcommand{\li}{{\lhd}}
\newcommand{\re}{{\rhd}}

\newcommand{\subtermeq}{\trianglelefteq}

\newcommand{\ZZ}{\mathbb{Z}}
\newcommand{\NN}{\mathbb{N}}

\newcommand{\Tiles}{\operatorname{\textsf{tiles}}}
\newcommand{\Tiled}{\operatorname{\textsf{tiled}}}
\newcommand{\Bord}{\operatorname{\textsf{bord}}}
\newcommand{\Btiled}{\operatorname{\textsf{btiled}}}
\newcommand{\Btiles}{\operatorname{\textsf{btiles}}}
\newcommand{\Forw}{\operatorname{\textsf{forw}}}
\newcommand{\Backw}{\operatorname{\textsf{backw}}}
\newcommand{\Lang}{\operatorname{\textsf{Lang}}}
\newcommand{\CC}{\operatorname{\textsf{CC}}}
\newcommand{\CCL}{\operatorname{\textsf{CC}}^\curvearrowleft}
\newcommand{\PA}{\operatorname{\textsf{Shift}}}
\newcommand{\PAL}{\operatorname{\textsf{Shift}}^\curvearrowleft}

\newcommand{\parto}{\operatorname{\rightharpoonup}} %

\newcommand{\Value}[1]{\llbracket#1\rrbracket}

\newcommand{\Alphabet}{\operatorname{\textsf{alphabet}}}

\newcommand{\Prefix}{\textsf{prefix}}
\newcommand{\Suffix}{\textsf{suffix}}

\newcommand{\shift}{\textsf{shift}}

\newcommand{\Term}{\textsf{Term}}
\newcommand{\Var}{\textsf{Var}}
\newcommand{\SN}{\textsf{SN}}

\newcommand{\OC}{\textsf{OC}}
\newcommand{\FC}{\textsf{FC}}

\newcommand{\ROC}{\textsf{ROC}}
\newcommand{\RFC}{\textsf{RFC}}

\newcommand{\TRFC}{\textsf{TRFC}}
\newcommand{\TROC}{\textsf{TROC}}

\newcommand{\TRFCU}{\textsf{TRFCU}}
\newcommand{\TROCU}{\textsf{TROCU}}

\newcommand{\Inf}{\textsf{Inf}}

\newcommand{\lhs}{\textsf{lhs}}
\newcommand{\rhs}{\textsf{rhs}}

\newcommand{\To}[1]{\stackrel{#1}{\to}}

\newcommand{\lab}{\textsf{lab}}
\newcommand{\aalg}{\mathcal{A}}
\newcommand{\dom}{\textsf{dom}}

\newcommand{\Deepee}{\stackrel{\textsf{DP}}{\to}}

\newcommand{\ignore}[1]{} %
\newcommand{\emptystring}{\epsilon}
\newcommand{\sem}[1]{\langle{#1}\rangle}
\newcommand\TTTT{%
 \textsf{T\kern-0.2em\raisebox{-0.3em}T\kern-0.2emT\kern-0.2em\raisebox{-0.3em}2}%
}

\newcommand{\tft}{\longrightarrow} %

\newcommand{\TileAllRFC}[1]{%
  \displaystyle%
  \operatornamewithlimits{\tft}^{\TRFC(#1)}
}
\newcommand{\TileRemoveRFC}[1]{%
  \displaystyle%
  \operatornamewithlimits{\tft}^{\TRFCU(#1)}
}

\newcommand{\TileAllROC}[1]{%
  \displaystyle%
  \operatornamewithlimits{\tft}^{\TROC(#1)}
}
\newcommand{\TileRemoveROC}[1]{%
  \displaystyle%
  \operatornamewithlimits{\tft}^{\TROCU(#1)}
}

\newcommand{\Weight}{\stackrel{\textsf{W}}{\tft}}
\newcommand{\Mirror}{\stackrel{\textsf{M}}{\tft}}

\newcommand{\Natural}{\NN}

\newcommand{\Mat}{\textsf{Mat}}
\newcommand{\Matrix}[2]{%
  \displaystyle%
  \operatornamewithlimits{\tft}^{\Mat(#2)}_{#1}
}

\newcommandx{\reminder}[2][1=]{}

\newcommand{\eex}{\qed} %

\begin{document}

\title{Sparse Tiling through  Overlap Closures %
  for Termination of String Rewriting 
}

\author[A.~Geser]{Alfons Geser}
\address{HTWK Leipzig, Germany}
\email{alfons.geser@htwk-leipzig.de}

\author[D.~Hofbauer]{Dieter Hofbauer}
\address{ASW -- Berufsakademie Saarland, Germany}
\email{d.hofbauer@asw-berufsakademie.de}

\author[J.~Waldmann]{Johannes Waldmann}
\address{HTWK Leipzig, Germany}
\email{johannes.waldmann@htwk-leipzig.de}

\newenvironment{theorem}{\thm}{\endthm}
\newenvironment{lemma}{\lem}{\endlem}
\newenvironment{corollary}{\cor}{\endcor}
\newenvironment{definition}{\defi}{\enddefi}
\newenvironment{proposition}{\prop}{\endprop}
\newenvironment{example}{\exa}{\endexa}
\newenvironment{algorithm}{\algo}{\endalgo}
\newenvironment{conjecture}{\conj}{\endconj}

\begin{abstract}
We over-approximate reachability sets in string rewriting
by languages defined by admissible factors, called tiles.
A sparse set of tiles contains only those
that are reachable in derivations,
and is constructed by completing an automaton.
Using the partial algebra defined by a sparse tiling
for semantic labelling, we obtain a transformational method
for proving local termination.
With a known result on forward closures,
and a new characterisation of overlap closures,
we obtain  methods for proving termination and relative termination, respectively.
We report on experiments showing the strength of these methods.
\end{abstract}

\maketitle

\section{Introduction}

Methods for proving  termination of rewriting (automatically)
can be classified
into syntactical (using a precedence on letters),
semantical (map each letter to a function on some domain),
or transformational, cf.~\cite{terese:termination}.
By applying a transformation, one hopes to
obtain an equivalent termination problem
that is easier to handle.

Another method for proving global termination, i.~e., termination of all derivations,
uses local termination, i.~e., termination of derivations
starting in a suitably restricted set of strings.
For example, termination on
the set $\RFC(R)$ of right-hand sides of forward closures of $R$
implies global termination.

Semantic labelling~\cite{DBLP:journals/fuin/Zantema95}
is a transformational termination proof method.
It was adapted to local termination via the concept
of partial model~\cite{DBLP:journals/corr/abs-1006-4955}.
In particular Section~8 in that paper
applies it to the RFC method, and states
that ``the challenge is to find a partial model such that the resulting
labelled total termination problem is easier than the original one.''
We now answer this challenge by Algorithm~\ref{imp:alg}
that computes a partial model in the \emph{$k$-shift algebra}.
Its domain consists of strings of length $k-1$.
The shift operation adds a letter at the right end,
and drops the letter at the left end.

If we semantically label a rewrite system $R$ over $\Sigma$
with respect to a $k$-shift algebra,
we obtain a rewrite system over $\Sigma^k$.
The elements of $\Sigma^k$ are the strings of fixed length $k$,
called \emph{tiles}.
The alphabet of the labelled system
contains factors (contiguous sub-strings)
of the original system's right-hand sides,
and possibly some more tiles due to a closure property.
Tiling is described in Section~\ref{sec:tiling}.

Even though our goal is termination of string rewriting,
we will occasionally use the language of terms
if it is necessary or convenient, for instance when we refer
to facts about partial algebras (Section~\ref{sec:pa}) or
when we represent a set of tiles by a deterministic automaton,
which is, in fact, a partial algebra (Section~\ref{sec:auto}).

We apply the tiling method to derivations
starting from right-hand sides of forward closures (Section~\ref{sec:imp})
and overlap closures (Section~\ref{sec:rel})
since local termination on these languages
implies global termination (a known result),
and relative termination, respectively.

Tiles encode information on adjacent letters,
a capability that may considerably increase
the power of other termination proof methods.
For instance, using tiling and weights only,
we obtain several automated termination proofs for Zantema's Problem
$\{a^2 b^2\to b^3a^3\}$, a classical benchmark, see Example~\ref{z001:ex}.
Our implementation is part of the Matchbox termination prover,
and it easily solves
several termination problems from the
Termination Problems Database\footnote{%
  The Termination Problems Database, Version~10.6,
  see~\url{http://termination-portal.org/wiki/TPDB}
}
that appear hard for other approaches,
e.~g., Examples~\ref{rbeans:ex}, ~\ref{r4:ex}, and ~\ref{w16:ex}.

Full, i.~e., non-sparse $2$-tiling has been employed by 
Jambox~\cite{jambox}
and Matchbox~\cite{DBLP:conf/rta/Waldmann04}
in the Termination Competition\footnote{%
  For the history and the results of the annual Termination Competition
  since 2004 see \url{http://termination-portal.org/wiki/Termination_Competition}
},
2006 followed in 2007 by Torpa~\cite{DBLP:journals/jar/Zantema05},
MultumNonMulta ~\cite{Hofbauer2016}, 
and \TTTT~\cite{DBLP:conf/rta/KorpSZM09}.
Sternagel and Middeldorp~\cite{DBLP:conf/rta/SternagelM08}
call \emph{root labelling}
the generalization of full $2$-tiling to term rewriting,
and combine it with the dependency pairs approach.
We note that  \emph{Self labelling}~\cite{MiddeldorpOhsakiZantema96}
can be seen as unrestricted shifting.

MultumNonMulta ranked first place in both categories
Standard and Relative String Rewriting of the Termination Competition 2018
mainly due to the use of (non-sparse) 2-tiling~\cite{Hofbauer2018}.
Sparse tiling contributed to Matchbox winning these two categories
of the Termination Competition 2019, see Section~\ref{sec:experiments}.

A preliminary version of this paper appeared as~\cite{DBLP:conf/rta/GeserHW19}
in the Proceedings of the 4th International Conference on Formal Structures for
Computation and Deduction, FSCD 2019.

\section{Motivating Examples}

We motivate our choice of the type of tiles for rewriting.

\subsection{Tiled Rewriting}
The 2-\emph{tiling} of the string $w=aaaba$ over alphabet $\Sigma= \{a,b\}$
is the string $\Tiled_2(w)=[aa, aa, ab, ba]$ over $\Sigma^2$,
where each letter corresponds to two adjacent letters over $\Sigma$. 
For the rewrite system $R=\{aa\to aba\}$ over $\Sigma$,
let $\Tiled_2^?(R)=\{[aa] \to [ab,ba]\}$ over $\Sigma^2$.
The question mark indicates that 
this is not the construction that we will actually use.
But let us pretend, and see what happens.
We then have $u\to_R v$ if and only if $\Tiled_2(u)\to_{\Tiled_2^?(R)}\Tiled_2(v)$,
so each $R$-derivation corresponds to a $\Tiled_2^?(R)$-derivation.
Also, we see that $aa$ disappears in $\Tiled_2^?(R)$,
therefore $\Tiled_2^?(R)$ is terminating, 
and we conclude that $R$ terminates as well.

\subsection{Context Closure}
Let us now consider $S=\{ab\to bbaa\}$.
Here, $\Tiled_2^?(S)=\{ [ab] \to [bb,ba,aa] \}$.
The letter (tile) $ab$ disappears, so $\Tiled_2^?(S)$ is terminating --- but
$S$ is not: There is an infinite derivation
$a\underline{ab} \to \underline{ab}baa \to bba\underline{ab}aa\to \dots$. 
By 2-tiling this derivation,
we obtain the derivation $[aa,ab]\to [ab,bb,ba,aa] \to \dots$
and already the first step is not represented by $\Tiled_2^?(S)$.

Naive tiling worked for $R=\{aa\to aba\}$
since the left-hand side and the right-hand side of that rule
have a common prefix and a common suffix of length 1.
We can make tiling work in general,
by padding rules with common prefixes and suffixes.
For $S$, we should use the 4 rules of
\[
  S' = \{[xa,ab,by]\to {}[xb,bb,ba,aa,ay] \mid x,y\in\{a,b\} \}.
\]
Then indeed each $S$-derivation corresponds to an $S'$-derivation
on tiled strings.

This method of tiling with context closure is a correct transformation
for proving termination, known as root labeling~\cite{DBLP:conf/rta/SternagelM08},
and it can be useful because it enlarges the signature.

\subsection{End Markers}

Now we want to tile derivations that start
in a given language $L\subseteq \Sigma^*$.
For instance, for $S=\{ab\to bbaa\}$ from above,
consider derivations starting with $w=aab$.
Here the approach does not work: The string $\Tiled_2(w)=[aa,ab]$
does not contain any redex for $S'$
as this set of redexes is $\{ [xa,ab,by]\mid x,y\in\{a,b\} \}$.
We repair this by introducing left and right end markers,
i.~e., new border symbols $\li$ and $\re$.
Then $\Bord_2(w)=\li a a b \re$
and $\Btiled_2(w)=[\li a,  aa, ab, b\re]$.
When we extend $S$ by contexts, we must also include end markers
and obtain the system with 9 rules
\[ \Btiled_2(ab\to bbaa)=
  \{ [xa,ab,by] \to [xb,bb,ba,aa,ay] \mid x\in\{\li,a,b\}, y\in\{a,b,\re\} \}.
\]
With this definition, we have restored the correspondence
between derivations of $R$ on $\Sigma^*$
and derivations of $\Btiled_2(R)$ on $\li\Sigma^*\re$.

In previous applications of 2-tiling (root labeling) for proving termination,
these end markers were not needed: if we have an infinite derivation,
then we can add contexts that contain letters from $\Sigma$
in order to obtain an infinite derivation that can be tiled.
Returning to our example,
the infinite $S$-derivation from $aab$
has no corresponding $S'$-derivation from $\Tiled_2(aab)$,
but we can add any right context, for instance, $a$,
to get an infinite $S$-derivation from $aaba$ that can be simulated by
an infinite derivation of $S'$ from $\Tiled_2(aaba)=[aa,ab,ba]$.

2-tiling had been applied after the dependency pairs transformation.
There it is important to consider rewrite steps at the top of a term,
that is, at one end of a string. In that case, end markers are already there,
in the form of marked top symbols.

\subsection{Sparse Tiling}
Adding contexts increases the number of rules substantially.
However, many of these rules are not needed in infinite derivations.
So we restrict to those tiled rules that use only tiles that appear in
derivations starting from $L$. This concept is called sparse tiling.
More formally, we want to use
\[
  \Btiled_T(R)= \Btiled_k(R)\cap T^*\times T^*
\]
such that $T$ can tile $L$, and $T$ is closed under $R$-rewriting.
This seems to have the problem that ``the set of tiles appearing
in some derivation'' is not computable.
However, we can compute an approximation.

Consider the problem of proving termination of
\[R=\{ba\to ac, cc\to bc, b\re\to ac\re, c\re\to bc\re\}\]
on the language $L=\{bc\re,ac\re\}$.
Let $T_0=\Btiles_2(L)= \{\li b, bc, c\re, \li a, ac, c\re\}$, the set
of tiles appearing in the right-bordered version 
$\{\li b c \re, \li a c \re\}$ of $L$. 
We consider all ways to cover the left-hand sides
of rules of $R$ by tiles from $T_0$.
For instance, $c\re \to bc\re $ can be left-extended by $a$ and $b$,
but not by $c$, since $cc\notin T_0$. We don't need right-extension
since the end marker is already present. 
So we obtain the two rules $[ac,c\re]\to [ab,bc,c\re]$
and $[bc,c\re]\to [bb,bc,c\re]$.
Note that $ab$ and $bb$ are fresh tiles,
appearing in a right-hand side, but not in $T_0$.
Now let $T_1=T_0\cup\{ab,bb\}$.
No further tiles can be added, so $T=T_1$ is a sparse set of tiles
for the given $R$ and $L$.

\subsection{Automata}
The next idea
is to represent a set of tiles $T$ by a deterministic automaton.
Then  ``what redexes can be $T$-tiled''
as well as ``what tiles are needed for the reduct''
are realized as tracing, or adding, paths in the automaton.
For example, in Figure~\ref{motiv:fig}, the left automaton represents $T_0$,
and the right automaton represents $T_1$.
\begin{figure}[ht!]
  \[    \begin{tikzpicture}[auto, on grid, node distance=20mm, 
  inner sep=2pt, semithick, >=stealth,  
  every state/.style={minimum size=20pt, inner sep=0pt, 
  initial distance=4mm, accepting distance=4mm}, 
  initial/.style=initial by arrow, initial text=, 
  accepting/.style=accepting by arrow] 
  \node (li) [state, initial] {$\li$};
  \node (b) [state, right of=li] {$b$};
  \node (a) [state, above of=li] {$a$};
  \node (c) [state, right of=a] {$c$};
  \node (re) [state, accepting, right of=c] {$\re$};
  \path[->] (li) edge node {$a$} (a);
  \path[->] (li) edge node [swap] {$b$} (b);
  \path[->] (a) edge node {$c$} (c);
  \path[->] (b) edge node [swap] {$c$} (c);
  \path[->] (c) edge node {$\re$} (re);
\end{tikzpicture} %
    \qquad    \begin{tikzpicture}[auto, on grid, node distance=20mm, 
  inner sep=2pt, semithick, >=stealth,  
  every state/.style={minimum size=20pt, inner sep=0pt, 
  initial distance=4mm, accepting distance=4mm}, 
  initial/.style=initial by arrow, initial text=, 
  accepting/.style=accepting by arrow] 
  \node (li) [state, initial] {$\li$};
  \node (b) [state, right of=li] {$b$};
  \node (a) [state, above of=li] {$a$};
  \node (c) [state, right of=a] {$c$};
  \node (re) [state, accepting, right of=c] {$\re$};
  \path[->] (li) edge node {$a$} (a);
  \path[->] (li) edge node [swap] {$b$} (b);
  \path[->] (a) edge node {$b$} (b);
  \path[->] (a) edge node {$c$} (c);
  \path[->] (b) edge node [swap] {$c$} (c);
  \path[->] (b) edge [loop right] node {$b$} (b);
  \path[->] (c) edge node {$\re$} (re);
\end{tikzpicture}
 \]
\caption{Automata for
  $T_0=\{\li b, bc, c\re, \li a, ac, c\re\}$ (left)
  and $T_1=T_0\cup\{ab,bb\}$ (right)}
\label{motiv:fig}
\end{figure}
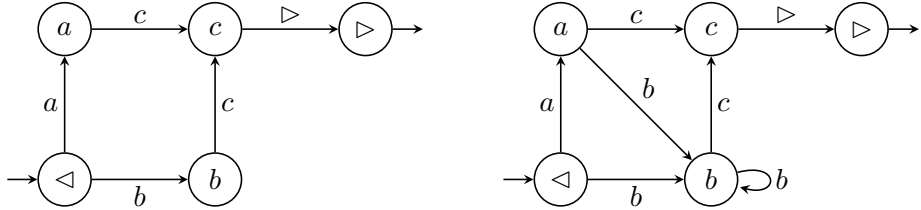

In all, we obtain a procedure for completing an automaton with respect to
a rewrite system---that is  guaranteed to halt
since the set of tiles of a fixed length is finite.

Automaton are deterministic, but not necessarily complete.
Indeed we hope they are sparse!

\section{Preliminaries}

Given a set of \emph{letters} $\Sigma$, called an \emph{alphabet},
a \emph{string} is a finite sequence of letters
over $\Sigma$. 
The number of its components is the \emph{length} of the string,
and the string of length zero, the \emph{empty string},
is denoted by $\emptystring$. 
If there is no ambiguity, we denote the string composed of the letters $a_1, \dots, a_n$
by $a_1 \dots a_n$. We deal, however, also with strings of strings,
and then use the list notation $[a_1,\dots, a_n]$.
Let $\Alphabet(w)$ denote the set of letters that occur in the string $w$.
By $\Prefix(S)$ and $\Suffix(S)$ we denote
the set of prefixes and suffixes, resp., of strings from the set $S$,
and by $\Prefix_k(S)$ and $\Suffix_k(S)$ we denote their restriction
to strings of length $k$. 

We use standard concepts and notation, see, e.~g., Book and Otto~\cite{BookOtto}. 
A \emph{string rewrite system} $R$ over an alphabet $\Sigma$
is a set of rewrite rules. 
It defines a rewrite relation $\to_R$ on $\Sigma^*$. 
For a relation $\rho$ on $\Sigma^*$ 
and a set $L \subseteq \Sigma^*$, 
let $\rho(L) = \{ y \mid \exists x \in L : (x,y) \in \rho\}$. 
Hence the set of \emph{$R$-reachable} strings from $L$ is ${\to_R^*}(L)$, 
or $R^*(L)$ for short.
A language $L \subseteq \Sigma^*$ is \emph{closed with respect to $R$}
if ${\to_R}(L) \subseteq L$. 

\begin{example}\label{reach:ex:1}
  For $R=\{ ba\to ac, cc\to bc, b\re \to ac\re, c\re \to bc\re \}$,
  the reachability set $R^*(\{ ac\re, bc\re \})$ equals $(a+b)b^* c\re$.
  By definition, this set is closed with respect to $R$.
  \eex
\end{example}

A rewrite system $R$ over $\Sigma$ is called
\emph{terminating on} $L\subseteq\Sigma^*$,
in symbols $\SN(R, L)$,
if for each $w\in L$, there is no infinite $R$-derivation starting at $w$,
and $R$ is called \emph{terminating}, written $\SN(R)$,
if there is no infinite $R$-derivation at all, i.~e.,
if $\SN(R, \Sigma^*)$.
For another rewrite system $S$ over $\Sigma$,
we say that $R$ is \emph{terminating relative to} $S$ \emph{on} $L$,
in symbols $\SN(R/S, L)$,
if for each $w\in L$, there is no infinite $(R\cup S)$-derivation starting at $w$
that has infinitely many $R$ rule applications.
By $R$ is \emph{terminating relative to} $S$,
in symbols $\SN(R/S)$, we mean $\SN(R/S, \Sigma^*)$. 
In~\cite{DBLP:journals/corr/abs-1006-4955},
termination on some language is called \emph{local} termination, 
in contrast to \emph{global} termination.

\section{Tiled Rewriting}\label{sec:tiling}

In order to approximate the language of reachable strings 
we consider prefixes, factors, and suffixes of fixed length, 
called \emph{tiles}. 
Left and right end markers $\li,\re\notin\Sigma$ allow for a uniform description. 
A similar formalization is employed for two-dimensional tiling in~\cite{hfl-2dl}.

\begin{definition}
  Let $\Gamma$ be an alphabet, and let $k \geq 1$.
  The factors of length $k$ of a given string over $\Gamma$ are called its
  \emph{$k$-tiles}.
  For $n\geq 0$ and $a_1,\dots,a_n \in \Gamma$,
  the \emph{$k$-tiled version} of the string $a_1\dots a_n$
  is the string over $\Gamma^k$ of all its $k$-tiles:
  \[ \Tiled_k(a_1\dots a_n) = [a_1\dots a_k, a_2\dots a_{k+1}, \dots, a_{n-k+1}\dots a_n] \]
  This string is empty in case $n < k$.
  Slightly more formally, 
  $\Tiled_k(w) = \emptystring$ for $|w| < k$, otherwise
  $\Tiled_k(a_1\dots a_n) = [t_1, \dots, t_{n-k+1}]$ 
  with $t_1 = \Prefix_k(w)$ and $t_{i+1} = \shift(t_i, a_{i+k})$, 
  using $\shift(w,a) = \Suffix_{|w|}(wa)$,
  i.~e., $t_i$ denotes the factor of length $k$ at position $i$. 
\end{definition}

\begin{definition}
  For $k \geq 0$ and $w\in\Sigma^*$, the \emph{$k$-bordered version} of $w$ is 
  $\Bord_k(w) = \li^{k}w\re^{k}$ 
  over $\Sigma \cup \{\li, \re\}$. 
  For $k \geq 1$ and $w\in\Sigma^*$,
  by $\Btiled_k(w)$ we abbreviate $\Tiled_k(\Bord_{k-1}(w))$,
  and $\Btiles_k(w)$ stands for $\Alphabet(\Btiled_k(w))$.
\end{definition}

\begin{example}
  $\Btiled_2(abbb) = \Tiled_2(\Bord_1(abbb))=
  \Tiled_2(\li abbb \re) = [\li a, ab, bb, bb, b\re]$,
  thus $\Btiles_2(abbb) = %
  \{\li a, ab, bb, b\re\}$.
  Further, 
  $\Btiles_2(\epsilon)=\{\li\re\}$,
  $\Btiles_2(a)=\{\li a, a\re\}$, 
  and 
  $\Btiled_3(a)=[\li\li a, \li a\re, a \re \re]$. 
  \eex
\end{example}

\begin{definition}\label{tiling:def:lang}
  For $k \geq 1$, the language defined by
  a set of tiles $T \subseteq \Btiles_k(\Sigma^*)$ is
 \[
  \Lang(T) =
    \{ w \in \Sigma^* \mid \Btiles_k(w) \subseteq T \} .
 \]
\end{definition}
This is a characterization of the class of 
strictly locally $k$-testable languages~\cite{CFA,ZALCSTEIN1972151},
a subclass of regular languages.

\begin{example}\label{tiling:ex}
  For $k = 2$ and $T = \{\li a, \li b, ab, ac, bb, bc, c\re \},$
  we obtain $\Lang(T) = (a+b)b^*c$.
  This is the language of Example~\ref{reach:ex:1}.
\eex
\end{example}

The central concept is the transformation 
from a rewrite system over $\Sigma$ 
to a rewrite system over tiles over $\Sigma$.
This is first defined  for the full set of tiles,
then for tiles from a subset.

\begin{definition}\label{trs:def:btiled}
  For a rule $\ell\to r$ over signature $\Sigma$
  we define a set of rules over signature $\Btiles_k(\Sigma^*)$ by
  \begin{align*}
    \Btiled_k(\ell\to r) &= 
    \{ \Tiled_k(x\ell y)\to \Tiled_k(xry) \mid 
    x \in \Tiles_{k-1}(\li^*\Sigma^*), 
    y \in \Tiles_{k-1}(\Sigma^*\re^*) \} , 
  \end{align*}
  and for a set of tiles $T \subseteq \Btiles_k(\Sigma^*)$ let
\[
 \Btiled_T(\ell\to r) = \Btiled_k(\ell\to r) \cap T^* \times T^*, 
\]
the set of tiled rules that use tiles from $T$ only.
Both $\Btiled_k$ and $\Btiled_T$ are extended to sets of rules.
\end{definition}

\begin{example}\label{trs:expl:btiled-k}
  The set 
  $\Btiled_2(ba \to ac)$ contains 16 rules, among them
  \begin{align*}
  [\li b,  ba, a \re] \to [\li a, ac, c \re],
  [\li b, ba, aa] \to [\li a, ac, ca], \dots, \\
  [ab, ba, a \re] \to [aa, ac, c \re], \dots, 
    [cb, ba, ac] \to [ca, ac, cc].
  \end{align*}
  For $S=\{ac,a\re,ba,bb,c\re\}$,
  we get $\Tiled_S(ba\to ac)=\{ [ bb,ba,a\re]\to [ba,ac,c\re] \}$.
  \eex
\end{example}

To obtain a correct method of proving termination of $R$ on $L$,
derivations of $R$ must be reflected faithfully in derivations of $\Btiled_T(R)$.
We need to construct $T$ in such a way
that $L$ is contained in $\Lang(T)$,
and $\Lang(T)$ is closed under $R$-rewriting (Theorem~\ref{trs:thm:main}).

\section{Tiles as Semantic Labels}\label{sec:pa}

We now present tiled rewriting as an instance of semantic labelling
with respect to a partial model~\cite{DBLP:journals/corr/abs-1006-4955}.
A set of tiles will be seen as a subset of the  \emph{$k$-shift algebra}.

To use concepts and results from local termination,
we need a translation to term rewriting.
We view strings as terms with unary symbols, 
a nullary symbol (representing $\emptystring$), and variables, 
where the rightmost (!) position in the string is the topmost position in the term.
As in~\cite{DBLP:journals/corr/abs-1006-4955}, we choose this order
(left to right in the string means bottom to top in the term)
since we later use deterministic automata,
working from left to right on the string,
realising bottom-up evaluation in the algebra. 
Thus a string $a_1 a_2\dots a_n$ is translated to 
the term $a_n(\dots a_2(a_1(z))\dots)$, 
denoted by $(z)a_1\dots a_n$ for short,
where $z$ is a variable symbol.
Analogously, we write $(\emptystring)a_1\dots a_n$ in case
the term is to be understood as a ground term.
This is just postfix notation for function application,
recommended also by Sakarovitch~\cite{ETA}, p.~12.

\subsection{Partial Algebras}

We recall concepts and notation %
from~\cite{DBLP:journals/corr/abs-1006-4955}.
For a signature $\Sigma$,
a \emph{partial $\Sigma$-algebra} $\aalg = (A, \Value{\cdot})$
consists of a non-empty set $A$
and for each $n$-ary $f\in \Sigma$ a partial function
$\Value{f}:A^n\parto A$.
Given $\aalg$ and a partial assignment of variables
$\alpha:X\parto A$, the \emph{interpretation} $\Value{t,\alpha}$
of $t\in\Term(\Sigma,X)$ is defined as usual,
but note that $\Value{t,\alpha}$ may also be undefined.
If $t$ is ground, we simply write $\Value{t}$.
A partial algebra is a \emph{partial model}
of a rewrite system $R$ if for each rewrite rule $(\ell\to r)\in R$,
and each partial assignment $\alpha:\Var(\ell)\parto A$,
definedness of $\Value{\ell,\alpha}$
implies $\Value{\ell,\alpha}=\Value{r,\alpha}$.

For a partial $\Sigma$-algebra $\aalg = (A, \Value{\cdot})$,
a term $t \in \Term(\Sigma, X)$, 
and a partial assignment $\alpha : X\parto A$,
let $\Value{t, \alpha}^*$
denote the set of defined values of subterms of $t$ under $\alpha$,
i.~e., $\{ \Value{s, \alpha} \mid s \subtermeq t \text{ and }
\text{$\Value{s, \alpha}$ is defined} \}$.
For $T \subseteq A$, let
$\Lang_{\aalg}(T)$ denote the set of ground terms
that can be evaluated inside $T$,
i.~e., $\{ t \in \Term(\Sigma) \mid \Value{t}^* \subseteq T \}$,
and let $\Lang_{\aalg} = \Lang_{\aalg}(A)$.
Note that a partial algebra is a deterministic tree automaton with set of states $A$,
and partiality means that the automaton may be incomplete.
The partial algebra $\aalg$ is \emph{core}
if each element is accessible, i.~e., is the value of a ground term.

The following obvious algorithm
computes the reachable subset of a model
by successively adding elements of the algebra
that become reachable via $R$-steps.

\begin{algorithm}\label{trs:algo:algebra}\hfill
  \begin{itemize}
    \item Specification:
    \begin{itemize}
\item Input: A term rewrite system $R$ over $\Sigma$,
  a finite  $\Sigma$-algebra $\aalg = (A, \Value{\cdot})$
  that is a model for $R$,
  a set $S \subseteq A$.   
\item Output: The minimal partial sub-algebra $(T,\Value{\cdot})$ of $\aalg$
  that contains $S$ and is a partial model for $R$.
\end{itemize}
\item Implementation:
  Let $T = \bigcup_i T_i$ for the sequence 
  $S = T_0 \subseteq T_1 \subseteq \cdots$  where 
  \[
    T_{i+1} = 
    T_i \cup \bigcup \{ \Value{r, \alpha}^* \mid
    (l \to r ) \in R, \alpha : \Var(\ell)\parto T_i, \Value{l, \alpha}^* \subseteq T_i \}
  \]
  where it is sufficient to compute a finite prefix of $S$.
  We return $T$. The valuation function can be inferred,
  as it is a restriction of $\aalg$.
\end{itemize}
\end{algorithm}

\begin{proof}
  Correctness: $(T,\Value{\cdot})$ is a partial model for $R$
  if for each rule $(\ell\to r)\in R$ and for each assignment $\alpha:\Var(\ell)\parto T$
  such that $\Value{\ell,\alpha}$ is defined,
  we have $\Value{\ell,\alpha}=\Value{r,\alpha}$.
  This property is ensured by construction.
  Termination: since the sequence $T_i$ is increasing with respect to $\subseteq$,
  and bounded from above by the finite set $A$, it is eventually constant.
\end{proof}

\subsection{Semantic Labelling.}

Each symbol is labelled by the tuple of the values of its arguments: 
For $t \in \Term(\Sigma, X)$ and $\alpha : \Var(t)\parto A$
such that $\Value{t, \alpha}$ is defined, 
the \emph{labelling $\lab_\aalg(t, \alpha)$ of $t$ with respect to $\alpha$} is
\begin{align*}
  \lab_\aalg(x, \alpha) &= x, \\
  \lab_\aalg(f(t_1, \dots, t_n), \alpha) &=
  f^{\Value{t_1,\alpha}, \dots, \Value{t_n,\alpha}}(\lab_\aalg(t_1, \alpha), \dots, \lab_\aalg(t_n, \alpha)),
\end{align*}
a term over the signature
$\lab_\aalg(\Sigma) = \{ f^\lambda \mid f \in \Sigma,
  \lambda \in A^{\textsf{arity}(f)}
  \text{ such that $\Value{f}(\lambda)$ is defined} \}$. 
For a term rewrite system $R$ over $\Sigma$ 
we define the \emph{labelling} of $R$ as the term rewrite system
$\lab_\aalg(R)$ over signature $\lab_\aalg(\Sigma)$ by 
\begin{align*}
  \lab_\aalg(R) &=
                  \{ \lab_\aalg(l,\alpha) \to \lab_\aalg(r,\alpha) \mid
                  (l \to r) \in R, \alpha : \Var(l) \parto A
                  \text{ such that } \Value{l, \alpha} \text{ is defined} \}
\end{align*}

\begin{theorem}\label{theo:pa}\cite[Theorem 6.4]{DBLP:journals/corr/abs-1006-4955}
  Let $R$ be a non-collapsing term rewrite system over $\Sigma$
  and $\aalg$ be  a core partial model for $R$.
  Then $R$ is terminating on 
  $\Lang_{\aalg}$ if and only if
  $\lab_\aalg(R)$ is terminating on $\Term(\lab_\aalg(\Sigma))$.
  \qed
\end{theorem}
A term rewrite system is non-collapsing if no right-hand side is a variable.
A string rewrite system is non-collapsing if no right-hand side is the empty string.
When we apply Theorem~\ref{theo:pa}, this  property will be ensured
by a context-closure operation.

We also need the following extension for relative termination:

\begin{theorem}\label{theo:pa-rel}
  Let $R$ and $S$ be non-collapsing term rewrite systems over $\Sigma$
  and $\aalg$ be  a core partial model for $R\cup S$.
  Then $R$ is terminating relative to $S$ on
  $\Lang_{\aalg}$ if and only if
  $\lab_\aalg(R)$ is terminating relative to $\lab_\aalg(S)$ on $\Term(\lab_\aalg(\Sigma))$.
\end{theorem}
\begin{proof}
  As in the proof of \cite[Theorem 6.4]{DBLP:journals/corr/abs-1006-4955}
  applied to $R\cup S$,
  keeping track of the origin ($R$ or $S$) of rules.
\end{proof}

\subsection{Tiled Rewrite Systems and Shift Algebras}\label{sec:trs}

We show that tiles define a partial algebra,
and when this algebra is used for semantic labelling,
we obtain a rewrite system over the alphabet of tiles.

\begin{definition}\label{trs:def:pa} 
  For $T\subseteq\Btiles_k(\Sigma^*)$,
  the partial algebra $\PA_k(T)$
  over signature $\Sigma\cup\{\epsilon,\re\}$
  has domain  %
  $\Prefix_{k-1}(T)\cup\Suffix_{k-1}(T)$,
  the interpretation of $\epsilon$ is $\li^{k-1}$,
  and each letter (unary symbol) $c\in \Sigma\cup\{\re\} $
  is interpreted by the unary partial function
  that maps $p$ to $\shift(p,c)$
  if $pc\in T$, and is undefined otherwise.
\end{definition}

We have the following obvious connection
(modulo the translation between strings and terms)
between the language of the algebra
(i.~e., all terms that have a defined value)
and the language of the set of tiles
(i.~e., all strings that can be covered):
\begin{proposition}\label{trs:prop:pa}
  For any set of $k$-tiles $T$, 
  $\Lang_{\PA_k(T)} = \Prefix(\Lang(T)\cdot\re^{k-1})$.
\end{proposition}

We need the prefix closure since a language of a partial algebra
is always subterm-closed,
according to the definition from~\cite{DBLP:journals/corr/abs-1006-4955},
a feature that
is criticised~\cite{DBLP:journals/iandc/FelgenhauerT17}.

A $k$-shift algebra is a model for a rewrite system $R$
if and only if $R$ does not change the $k-1$ topmost symbols.
This property can be guaranteed by the following closure operation. 

\begin{definition}\label{trs:def:cc}
  For $k \ge 1$ and a string rewrite system $R$ over $\Sigma$,
  define its \emph{context closure}, the
  term rewrite system $\CC_k(R)$ over $\Sigma\cup\{\emptystring, \re\}$, 
  where $\emptystring$ is a constant, all other symbols are unary,
  and $z$ is a variable symbol as
  \begin{align*}
    \CC_k(R) ={} &\{ (z)\ell y \to (z)ry \mid (\ell\to r)\in R, y\in \Tiles_{k-1}(\Sigma^*\re^*) \} . 
  \end{align*}
\end{definition}

Rewrite steps of $R$ on $\Sigma^*$ 
are directly related to term rewrite steps of the context closure of $R$
on (the set of terms corresponding to) $\Sigma^*\re^{k-1}$:

\begin{proposition}
  $s \to_R t$ if and only if
  $(\emptystring)s\re^{k-1}\to_{\CC_k(R)}(\emptystring)t\re^{k-1}$.
\end{proposition}

Since $\CC_k(R)$ does keep the $k-1$ topmost (rightmost) symbols intact,
the shift algebra of $T$ is a partial model
provided it contains a sufficiently large set of tiles:

\begin{proposition}\label{trs:prop:pm}
  For a set of $k$-tiles $T$ and a rewrite system $R$,
  if   $\Lang_{\PA_k(T)}$ is closed with respect to $R$,
  then  $\PA_k(T)$  is a core partial model for $\CC_k(R)$.
\end{proposition}

Given a partial model, we use it for semantic labelling.
The labelling of $\CC_k(R)$ with respect to $\PA_k(T)$
produces a term rewrite system
that can be re-transformed to a string rewrite system
by replacing each function symbol $c$,
that is labelled with an element $p$ from the algebra,
to the string (the tile) $pc$.

We highlight the connection between tiling and semantic labelling.
\begin{proposition}\label{trs:prop:btiled-lab}
  Let $\aalg=\PA_k(T)$. Let $x\in \dom(\aalg)$,
  that is, $x\in T\cup\{\li^{k-1}\}$,
  and let $\alpha$ denote the assignment $z\mapsto x$.
  Then $\Tiled_T(xwy)=\lab_\aalg((z)w y,\alpha)$.
\end{proposition}
As this proposition switches tacitly between strings (in the left-hand
side) and terms (in the right-hand side),
we make this more explicit in  the proof.
  \begin{proof}
  The symbol at position $i$ 
  in the term $(z)w y$
  (counting from the left, i.~e., from the bottom)
  is labelled with the interpretation of $xp$,
  where $p$ denotes the prefix (subterm) of length $i-1$.
  By the semantics of the shift operation,
  this value is the suffix of length $k-1$ of $xp$.
\end{proof}

\begin{proposition}\label{trs:prop:faith}
  $\Btiled_T(R)$ is exactly the (string rewriting translation of the)
  labelling of $\CC_k(R)$ with respect to $\PA_k(T)$.
\end{proposition}
\begin{proof}
  For a rule $(\ell\to r)\in R$ and $y\in\Tiles_{k-1}(\Sigma^*\re^*)$, 
  we have $((z)\ell y\to (z)ry) \in \CC_k(R)$. This will be labelled for
  all $\alpha:z\mapsto x$ for $x\in \dom(\aalg)$.
  Apply Proposition~\ref{trs:prop:btiled-lab} to the left-hand sides
  and right-hand sides of labelled rules.
\end{proof}

To actually enumerate $\Btiled_T(R)$ in an implementation,
we will fuse both parts of Definition~\ref{trs:def:btiled}
by restricting contexts $x$ and $y$ to be elements of $T^*$ right from the beginning.

\begin{theorem}\label{trs:thm:main}
For $k \geq 1$ and $T\subseteq \Btiles_k(\Sigma^*)$, 
if $\Lang(T)$ is closed with respect to $R$, 
then  $R$ is terminating on $\Lang(T)$ 
if and only if $\Btiled_T(R)$ is terminating.
\end{theorem}

\begin{proof}
  The interesting case is $k\ge 2$.
  Then $\CC_k(R)$ is non-collapsing,
  and the claim follows  
  by Proposition~\ref{trs:prop:faith}
  and Theorem~\ref{theo:pa},
  applicable due to Proposition~\ref{trs:prop:pm}.

  For $k=1$, we have $\CC_1(R)=R$, which is collapsing
  in case $R$ contains a rule $\ell\to\epsilon$.
  But $\Btiles_1(\Sigma^*)$ is $\Sigma$,
  and $\Lang(T)$ is $\Gamma^*$ for some $\Gamma\subseteq \Sigma$.
  Then $\Btiled_T(R)$ is the restriction of $R$ to $\Gamma$,
  and the result follows.
\end{proof}

In applications (Algorithms~\ref{imp:alg} and~\ref{rel-imp:alg}),
the algebra (the set $T$) will be closed under other operations as well,
but we are finally interested in the tiling for $R$ only.
This already happens in the following example.

\begin{example}\label{trs:expl:rfc}%
  Let
  \[ R=\{ ba\to ac, cc\to bc, b\re \to ac\re, c\re \to bc\re \}, \]
  and let $L = R^*(\{ ac\re, bc\re \}) = (a+b)b^* c\re$,
  see Example~\ref{reach:ex:1}.
  Then $L = \Lang(T)\re$ for the set of tiles 
  \[ T = \{ \li a, \li b, ab, ac, bb, bc, c\re \}, \]
  see Example~\ref{tiling:ex}.
  By definition, $L$ is closed with respect to $R$,
  so $L$ is also closed with respect to
  \[ R_0=\{ba\to ac,cc\to bc\}\subset R. \]
  On the other hand, $\Tiled_S(R_0)$ is empty,
  as the left-hand sides cannot be covered:
  each left-hand side of  $\Tiled_2(R_0)$
  contains the letter (the tile) $ba$, or the tile $cc$,
  but neither of them is in $S$.
  This means that $\Btiled_S(R_0)$ is trivially terminating.
  By Theorem~\ref{trs:thm:main}, $R_0$ is terminating on $L$. 
  See Example~\ref{auto:expl:rfc} for a computation that
  produces $T$ from $R_0$.
  \eex
\end{example}

Theorem~\ref{trs:thm:main} can be extend for relative termination:
\begin{theorem}\label{trs:thm:rel}
For $k \geq 1$ and $T\subseteq \Btiles_k(\Sigma^*)$, 
if $\Lang(T)$ is closed with respect to $R\cup S$, 
then  $R$ is terminating relative to $S$ on $\Lang(T)$ 
if and only if $\Btiled_T(R)$ is terminating relative to $\Btiled_T(S)$.
\end{theorem}
\begin{proof}
  By Proposition~\ref{trs:prop:faith}
  and Theorem~\ref{theo:pa-rel},
  applicable due to Proposition~\ref{trs:prop:pm}.
\end{proof}
\section{Completion of Shift Automata}\label{sec:auto}

To apply Theorem~\ref{trs:thm:main},
we obtain an $R$-closed set $T$ of tiles
from Algorithm~\ref{trs:algo:algebra}.
Those results were presented for terms, not strings.
The detour via terms was taken in order to obtain a correctness proof.
But it has another virtue: the partial algebra $\PA_k(T)$
is in fact a deterministic automaton,
and  Algorithm~\ref{trs:algo:algebra}
is in fact a method of completing an automaton with respect
to a rewrite system, cf.~\cite{DBLP:conf/rta/Genet98}.

This model is useful both for understanding the method
(we use drawings of automata in the following examples),
and for implementing it:
we use the transition relation of
the automaton when we  check whether a redex, or reduct,
is covered by the (current) set of tiles.
A naive implementation of coverage by a set of tiles would be costly.
A more clever implementation could use suffix trees,
which are in fact automata~\cite{DBLP:journals/tcs/BlumerBHECS85}.
Our implementation represents an automaton
as a collection of sparse transition matrices~\cite{DBLP:journals/corr/Waldmann16}.
\begin{definition}
  For $k \geq 1$, a finite deterministic automaton $A$ over alphabet $\Gamma$
  is called a \emph{$k$-shift automaton}
  if the set of states is a subset of $\Gamma^{k-1}$,
  and for each transition $p\stackrel{c}{\to}_Aq$,
  state $q$ is the suffix of length $k-1$ of $pc$.
  This automaton represents the set of tiles (of length $k$)
  \( \Tiles(A) = \{ pc \mid p\stackrel{c}{\to}_A q \} .\)
\end{definition}

Each set $T\subseteq \Gamma^k$, and $w\in\Gamma^{k-1}$,
uniquely determine a $k$-shift automaton $A$
with $\Tiles(A)=T$ and minimal set of states $Q=\Prefix_{k-1}(T)\cup\Suffix_{k-1}(T)$,
start state $w$, and all states from $Q$ as accepting states.

The full $k$-shift automaton, which has $T=\Gamma^k$,
is the \emph{k-local universal automaton} of Perrin~\cite[p.~27]{Perrin90}.
It is called the \emph{subword automaton} in~\cite[p.~265]{GenovaHoogeboom17}. 

Each $k$-shift automaton $A$ over $\Gamma$
is a partial algebra $\aalg$ over $\Gamma$,
where the letter $c\in \Gamma$ is interpreted
by the partial function that maps $p$ to $\shift(p,c)$
in case that the result is a state of $A$.

Condition $\Value{\ell, \alpha}^* \subseteq T_i$
of Algorithm~\ref{trs:algo:algebra}
is equivalent to the existence of a path in the automaton $T_i$
that starts at state $p=\alpha(z)$ and is labelled $\ell$.
We call this a \emph{redex path} $p\stackrel{\ell}{\to} q$.
Adding tiles then corresponds to adding edges and states.
Whenever we add edges for some reduct path $p\stackrel{r}{\to}q'$,
corresponding to $\Value{r, \alpha}^* \subseteq T_i$,
the target state of each transition is determined
by the shift property of the automaton.
This is in contrast to other completion methods
where there is a choice of adding fresh states,
or re-using existing states.
The set of states could be defined to be $\Gamma^{k-1}$ in advance,
but for efficiency, we %
only store accessible states,
and %
add states as soon as they become accessible.

With the automata representation,
we implement $\Btiled_T(R)$ as follows:
To determine  $x\ell y$ in Definition~\ref{trs:def:btiled},
we compute all pairs $p,q$ of states with $p\stackrel{\ell}{\to}q$.
This can be done by starting at each $p$,
but our implementation uses the
product-of-relations method of~\cite{DBLP:journals/corr/Waldmann16}.
Note that $p$, the state where the redex path starts,
\emph{is} actually $x$, the left context.
From state $q$, we follow all paths of length $k-1$
to determine the set of $y$ (right contexts).
For each such pair $(x,y)$, we add the path
starting at $x$ labelled $ry$.
Note that this path (for the context-closed reduct)
meets the path for $\ell y$ (the context-closed redex) in the end,
since the automaton is a shift automaton.
The tree search for possible $y$ can be cut short
if we detect earlier that these paths meet.

The following example demonstrates completion only.
For examples that use the completed automaton
for semantic labelling, see Section~\ref{sec:imp}.

\begin{example}\label{auto:expl:rfc}
  For $R=R_0\cup R_1$ with
  \[ R_0=\{ bay\to acy, ccy\to bcy\mid y\in \{a,b,c,\re\}\},
    R_1=\{b\re \to ac\re, c\re \to bc\re \},
  \]
  we are interested in the reachability set $R^*(L)$
  for $L=\{ac\re,bc\re\}$.
  We choose $k = 2$ and represent
  $\Tiled_2(L) = \{ [\li b, b c, c \re], [\li a, a c, c \re] \}$
  by the left automaton in Figure~\ref{auto:fig:rfc}.

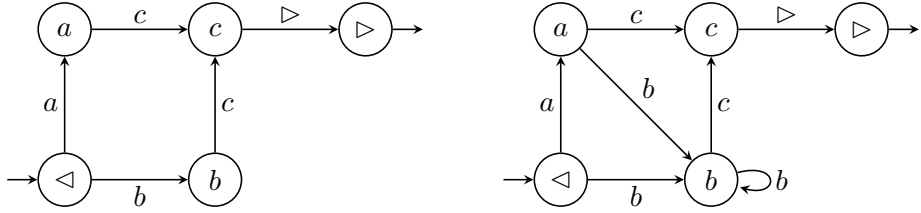
\begin{figure}[ht!]
  \[ \begin{tikzpicture}[auto, on grid, node distance=20mm, 
  inner sep=2pt, semithick, >=stealth,  
  every state/.style={minimum size=20pt, inner sep=0pt, 
  initial distance=4mm, accepting distance=4mm}, 
  initial/.style=initial by arrow, initial text=, 
  accepting/.style=accepting by arrow] 
  \node (li) [state, initial] {$\li$};
  \node (b) [state, right of=li] {$b$};
  \node (a) [state, above of=li] {$a$};
  \node (c) [state, right of=a] {$c$};
  \node (re) [state, accepting, right of=c] {$\re$};
  \path[->] (li) edge node {$a$} (a);
  \path[->] (li) edge node [swap] {$b$} (b);
  \path[->] (a) edge node {$c$} (c);
  \path[->] (b) edge node [swap] {$c$} (c);
  \path[->] (c) edge node {$\re$} (re);
\end{tikzpicture} %
 \qquad \begin{tikzpicture}[auto, on grid, node distance=20mm, 
  inner sep=2pt, semithick, >=stealth,  
  every state/.style={minimum size=20pt, inner sep=0pt, 
  initial distance=4mm, accepting distance=4mm}, 
  initial/.style=initial by arrow, initial text=, 
  accepting/.style=accepting by arrow] 
  \node (li) [state, initial] {$\li$};
  \node (b) [state, right of=li] {$b$};
  \node (a) [state, above of=li] {$a$};
  \node (c) [state, right of=a] {$c$};
  \node (re) [state, accepting, right of=c] {$\re$};
  \path[->] (li) edge node {$a$} (a);
  \path[->] (li) edge node [swap] {$b$} (b);
  \path[->] (a) edge node {$b$} (b);
  \path[->] (a) edge node {$c$} (c);
  \path[->] (b) edge node [swap] {$c$} (c);
  \path[->] (b) edge [loop right] node {$b$} (b);
  \path[->] (c) edge node {$\re$} (re);
\end{tikzpicture}
 \]
\caption{Completing a 2-shift automaton}
\label{auto:fig:rfc}
\end{figure}

  In the initial automaton we look for paths of the form $p\stackrel{\ell}{\to}q$
  for some rule $\ell\to r \in R$.
  Two such paths exist, $a\stackrel{c\re}{\to}\re$ and $b\stackrel{c\re}{\to}\re$. 
  Completion therefore adds the paths
  $a\stackrel{bc\re}{\to}\re$ and $b\stackrel{bc\re}{\to}\re$ for the corresponding
  right-hand sides, resulting in the new edges $a\stackrel{b}{\to}b$ and 
  $b\stackrel{b}{\to}b$ (and no new nodes),
  depicted by the right automaton $A$,
  with   \[\Tiles(A) = \{ \li a, \li b, ab, ac, bb, bc, c\re \}. \]
  No further completion steps are possible, thus
  $R^*(\{ ac\re, bc\re \}) \subseteq \Lang(\Tiles(A))\re$ .
\eex
\end{example}

\section{Termination Proofs via Forward Closures}\label{sec:imp}

We recall the method of proving global termination of $R$
by proving local termination on the set $\RFC(R)$
of right-hand sides of forward closures,
and then use tiling to both approximate that set, and label $R$.
We obtain a transformational termination proof method.
Indeed, we show examples that apply the RFC tiling transformation repeatedly.
This is in contrast to RFC matchbounds~\cite{DBLP:journals/aaecc/GeserHW04}
which use the same idea, but as a one-shot method.

\subsection{Forward Closures}

Given a rewrite system $R$ over alphabet $\Sigma$,
a \emph{closure} $C=(s,t)$ of $R$
is a pair of strings with $s\to_R^+ t$
such that each position between letters of $s$
was touched by some step of the derivation. 
In particular, we use \emph{forward closures}~\cite{LankMuss-78}.
The set $\FC(R)$ of forward closures of $R$ is defined as the least set
of pairs of strings that contains $R$ and satisfies
\begin{enumerate}
\item
  if $(s,xuy)\in\FC(R)$ and $(u,v)\in\FC(R)$ then $(s,xvy)\in\FC(R)$, 
\item
  if $(s,xu)\in\FC(R)$ and $(uy,v)\in\FC(R)$ for $u \neq \epsilon \neq y$
  then $(sy,xv)\in\FC(R)$.
\end{enumerate}
Let $\RFC(R)$ denote the set $\rhs(\FC(R))$
of right-hand sides of forward closures of $R$

They are related to termination by
\begin{theorem}\label{theo:rfc}~\cite{DBLP:conf/icalp/Dershowitz81}
  $R$ is terminating on $\Sigma^*$
  if and only if $R$ is terminating on $\RFC(R)$.
\end{theorem}

The set $\FC(R)$ can also be characterized without recursion in the second
premise, as observed by Herrmann~\cite[Corollaire~2.16]{Her-Habil}
in the term rewriting case:
\begin{enumerate}
\item
  if $(s,x\ell y)\in\FC(R)$ and $(\ell,r)\in R$ then $(s,xry)\in\FC(R)$, 
\item
  if $(s,x\ell_1)\in\FC(R)$ and $(\ell_1\ell_2,r)\in R$ for $\ell_1 \neq \epsilon \neq \ell_2$
  then $(s\ell_2,xr)\in\FC(R)$.
\end{enumerate}

This can be used to recursively characterize the set $\RFC(R) = \rhs(\FC(R))$
of right hand sides of forward closures
directly~\cite{DBLP:journals/aaecc/GeserHW04}:
\begin{enumerate}
\item
  if $x\ell y\in\RFC(R)$ and $(\ell,r)\in R$ then $xry\in\RFC(R)$, 
\item
  if $x\ell_1\in\RFC(R)$ and $(\ell_1\ell_2,r)\in R$ for $\ell_1 \neq \epsilon \neq \ell_2$
  then $xr\in\RFC(R)$.
\end{enumerate}

The set $\RFC(R)$ can also be characterized by rewriting,
where the fresh symbol $\re \notin \Sigma$ restricts rewriting to suffixes.

\begin{definition} For $k\ge 1$,
  let $ \Forw_k(R)$ denote
  $\{ \ell_1\re^k \to r\re^k \mid (\ell_1\ell_2\to r)\in R, \ell_1 \neq \epsilon \neq \ell_2\},$
  and abbreviate $\Forw_1(R)$ by $\Forw(R)$.
\end{definition}
\begin{proposition}\label{prop:rfc}
  $\RFC(R)\re = (R \cup \Forw(R))^*(\rhs(R)\re)$.
\end{proposition}

For a self-contained proof of a similar result,
see \cite[Section 6]{DBLP:journals/jar/Zantema05}.
The difference is that our $\Forw(R)$  checks the end marker and keeps it,
while \cite{DBLP:journals/jar/Zantema05} checks and removes it,
starting from $\rhs(R)\re^*$.

\begin{example}\label{closure:ex:1}
  For $R = \{ba\to ac, cc\to bc \}$ we have
  $\Forw_1(R) = \{b\re\to ac\re, c\re\to bc\re \}$ and
  $\RFC(R) = (a+b) b^* c$, cf.~Example~\ref{reach:ex:1}.
  As $\RFC(R)$ contains no $R$-redex, 
  $R$ is trivially terminating on $\RFC(R)$, 
  therefore $R$ is terminating by Theorem~\ref{theo:rfc}.   
  Later, we apply Algorithm~\ref{imp:alg}
to obtain this termination proof automatically.
\eex
\end{example}

In the previous example, $\RFC(R)$ was a regular language.
Things are not always that simple:

\begin{example}\label{closure:ex:nonreg}
  For $R = \{a \to bab\}$,
  we have $\RFC(R) = \ROC(R) = \{b^n a b^n \mid n\geq 1\}$,
  a non-regular language. 
  Also for the terminating system $R = \{ab \to baa\}$,
  the language $\RFC(R)\cap b^*a^* = \{b^n a^{2^n} \mid n\geq 1\}$
  is non-regular, hence also $\RFC(R)$ is non-regular
  (this is Example~19 from~\cite{DBLP:journals/aaecc/GeserHW04}).
  \eex
\end{example}

This implies that for any algorithm that computes a finite automaton $A$
that contains $\RFC(R)$, there are some inputs $R$
where the inclusion $\RFC(R)\subset \Lang(A)$ is strict.

\subsection{Tiling for Forward Closures}

We transform a termination problem as follows:

\begin{algorithm}\label{imp:alg}[Tiling for RFC, Abbreviation $\TRFC$]
  \begin{itemize}
  \item Specification:
    \begin{itemize}
    \item Input: A rewrite system $R$ over $\Sigma$, a number $k$
    \item Output: A rewrite system $R'$ over $\Btiles_k(\Sigma^*)$
      such that $\SN(R)\iff\SN(R')$
    \end{itemize}
  \item Implementation:
    We call Algorithm~\ref{trs:algo:algebra},
    with these arguments:
    \begin{enumerate}
    \item 
      the term rewrite system $\CC_k(R) \cup \Forw_k(R)$
      over signature  $\Sigma \cup \{\re\}$,

    \item
      the $k$-shift algebra over signature $\Sigma\cup\{\epsilon,\re\}$,
    \item 
      and the set of domain elements $\Prefix_{k-1}(T)\cup\Suffix_{k-1}(T)$
      for $T=\Btiles_k(\rhs(R))$.
    \end{enumerate}
    Algebra operations are implemented as in Section~\ref{sec:auto}.
    We obtain a partial algebra $\aalg=(A,\Value{\cdot})$.
    We output $R'=\Btiled_U(R)$, where $U=\Tiles(\aalg)$.
  \end{itemize}
\end{algorithm}
\begin{proof} Correctness:
  $\Lang(T)$ contains $\rhs(R)$ by construction:
  for any $r\in\rhs(R)$, we have  $\Tiles_k(\li^{k-1}r\re^{k-1}) \subseteq T$
  by construction of $T$.
  Proposition~\ref{trs:prop:btiled-lab} for $\alpha:z\mapsto\li^{k-1}$
  implies that  $\lab_\aalg((\epsilon)r\re^{k-1})$ is defined.

  $\PA_k$ is  a model for $\CC_k(R)\cup\Forw_k(R)$:
  for each $(l,r)\in\CC_k(R)\cup\Forw_k(R)$,
  we have $\Suffix_{k-1}(l)=\Suffix_{k-1}(r)$,
  therefore, for each $\alpha$, $\Value{l,\alpha}_{\PA_k}=\Value{r,\alpha}_{\PA_k}$.

  This means that the precondition of Algorithm~\ref{trs:algo:algebra} holds.
  Then $\aalg$ is a partial model of $\CC_k(R)\cup\Forw_k(R)$
  that contains $\rhs(R)\re^{k-1}$.

  So $\Lang(\aalg)$ contains $(\CC_k(R)\cup\Forw_k(R))^*(\rhs(R)\re^{k-1})$,
  a superset of $\RFC(R)\re^{k-1}$ by Proposition~\ref{prop:rfc}.

  By Theorem~\ref{trs:thm:main},   $\SN(R')$ iff $\SN(R,\Lang(U))$. 

  By  $\RFC(R)\subseteq\Lang(U)$,
  we have that $\SN(R,\Lang(U))$ implies $\SN(R,\RFC(R))$.

  By Theorem~\ref{theo:rfc}, this implies $\SN(R)$.

  For the other direction, $\SN(R)$ implies $\SN(R,L)$
  for any language $L$, in particular, for $L=\Lang(U)$.
\end{proof}

This approach had already been described
in~\cite{DBLP:journals/corr/abs-1006-4955}, Section 8, 
but there it was left open how to find a suitable partial algebra.
An implementation used a finite-domain constraint solver,
but then only small domains could be handled.

In the present paper, we instead construct
a suitable $k$-shift algebra by completion.
Even if it that algebra is large,
it might help solve the termination problem,
cf.~Example~\ref{z001:ex} below.
We give a few smaller examples first.
In fact, Example~\ref{auto:expl:rfc} already illustrates the algorithm,
since $R_0=\CC_1(\{ba\to ac,cc\to bc\})$, $R_1=\Forw_1(\{ba\to ac,cc\to bc\})$,
and $L=\rhs(\{ba\to ac,cc\to bc\})\re$.

\begin{example}

  We apply Algorithm~\ref{imp:alg} with $k=3$ to $R=\{ab^3\to bbaab\}$.
  We obtain 11 reachable tiles and 12 labelled rules.
  All of them can be removed by weights. 
  We start with the automaton for $\Btiled_3(bbaab)$
  (solid edges in Figure~\ref{ex4:fig}).
\begin{figure}[ht!]
  \begin{tikzpicture}[auto, on grid, node distance=20mm, 
  inner sep=2pt, semithick, >=stealth,  
  every state/.style={minimum size=20pt, inner sep=0pt, 
  initial distance=4mm, accepting distance=4mm}, 
  initial/.style=initial by arrow, initial text=, 
  accepting/.style=accepting by arrow]
  
  \node(ll) [state,initial] {$\li^2$} ;
  \node(lb) [state,right of=ll] {$\li b$} ;
  \node(bb) [state,right of=lb] {$b^2$} ;
  \node(ba) [state,below of=bb] {$ba$} ;
  \node(aa) [state,right of=ba] {$a^2$} ;
  \node(ab) [state,above of=aa] {$ab$} ;
  \node(br) [state,right of=ab] {$b\re$} ;
  \node(rr) [state,accepting,right of=br] {$\re^2$} ;
  \path[->] (ll) edge node {$b$} (lb);
  \path[->] (lb) edge node {$b$} (bb);
  \path[->] (bb) edge [swap] node {$a$} (ba);
  \path[->] (ba) edge [swap] node {$a$} (aa);
  \path[->] (aa) edge [swap] node {$b$} (ab);
  \path[->] (ab) edge node {$\re$} (br);
  \path[->] (br) edge node {$\re$} (rr);
  \path[->,dashed] (ba) edge [bend right] node {$b$} (ab);
  \path[->,dashed] (ab) edge [swap] node {$b$} (bb);
  \path[->,dotted] (bb) edge [loop above] node {$b$} (bb);
  \path[->,dash dot] (ab) edge [bend right] node {$a$} (ba);
\end{tikzpicture}
\caption{Algorithm $\TRFC_3$ on input $\{ab^3\to bbaab\}$}
\label{ex4:fig}
\end{figure}
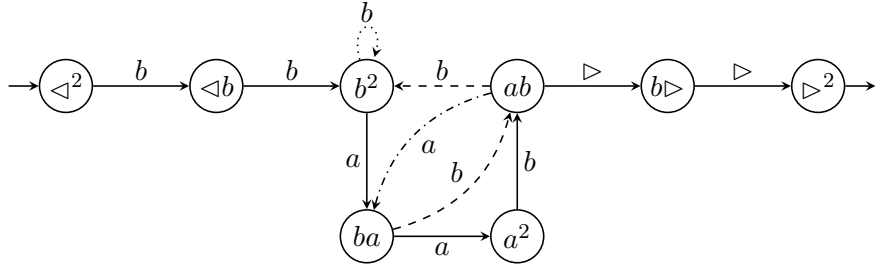

  It contains no $R$-redex. There is a $\Forw_r(R)$-redex
  for $ab\re^2\to bbaab\re^2$ starting at $ba$.
  We add a reduct path, starting with two fresh (dashed) edges.
  This creates a $\Forw_3(R)$-redex for $ab\re^2\to bbaab\re^2$  from $b^2$.
  To cover this, we add the loop at $b^2$ (dotted).
  Now we have $\CC_3(R)$-redexes
  $ba \To{a} a^2 \To{b} ab \To{b} b^2 \To{b} b^2\To{x}bx\To{y}xy $
  for $x,y\in\{a,b\}$.
  The corresponding reduct paths are
  $ba\To{b} ab\To{b} b^2\To{a} ba\To{a} a^2\To{b} ab\To{x}bx\To{y}xy$,
  for which  we add one more edge $ab\To{a} ba$ (dash-dotted),
  as $ab\To{b} b^2$ is already present.
  This introduces $\CC_3(R)$-redexes from $ab$.
  Their reduct paths are already present. 
  The automaton is now closed with respect to $\CC_3(R)\cup\Forw_3(R)$.
  It represents the set of tiles
  \[
    T=\{\li\li b,\li bb,bba, bbb,
    baa,bab,aab,aba,abb,ab\re,b\re\re\}.
  \]
  Absent from $T$ are
  \begin{itemize}
  \item 
  $\li\li\re,\li\re\re,\li\Sigma\re$
  (meaning that $\RFC(R)$ does not contain strings of length 0 or 1), 
\item  as well as $\li a \Sigma,\li ba, \Sigma a\re$
  (meaning that $\RFC(R)$ starts with $b^2$ and ends with $b$),
\item  and $a^3$ (meaning that $\RFC(R)$ does not have $a^3$ as a factor).
  \end{itemize}
  
  Finally, we compute $\Btiled_T(R)$.
  There are three $R$-redex paths in the automaton,
  starting at $b^2,ba, ab$, respectively, and all ending in $b^2$.
  Then $\CC_3(R)$ has $3\times 2^2=12$ redexes,
  resulting in 12 tiled rules, where $x,y\in\Sigma$:
\begin{align*}
  [bba,bab,abb,b^3,bbx,bxy] &\to [b^3,b^3,bba,baa,aab,abx,bxy] \\
  [baa,aab,abb,b^3,bbx,bxy]& \to [bab,abb,bba,baa,aab,abx,bxy] \\
  [aba,bab,abb,b^3,bbx,bxy]& \to [abb,b^3,bba,baa,aab,abx,bxy]   
\end{align*}
   With the following weights, all rules are strictly decreasing:
   \[ bbb\mapsto 8, bab\mapsto 4, abb\mapsto 3,  bba\mapsto 3, 
     \text{others}\mapsto 0.
   \]
   This shows termination of $\Btiled_T(R)$, thus, of $R$.
\eex
 \end{example}

\begin{example}
  We apply Algorithm~\ref{imp:alg}
  with tiles of width $k=4$  to $R=\{a^3\to a^2b^3a^2\}$.
  It turns out that the labelled system has just one rule
  that can be removed by counting letters. 
  We start with the automaton for $\rhs(R)$ (the solid arrows
  in Figure~\ref{ex3:fig}).
\begin{figure}[ht!]
  \begin{tikzpicture}[auto, on grid, node distance=16mm, 
  inner sep=2pt, semithick, >=stealth,  
  every state/.style={minimum size=20pt, inner sep=0pt, 
  initial distance=4mm, accepting distance=4mm}, 
  initial/.style=initial by arrow, initial text=, 
  accepting/.style=accepting by arrow]
  
  \node(lll) [state,initial] {$\li^3$} ;
  \node(lla) [state,right of=lll] {$\li^2a$} ;
  \node(laa) [state,right of=lla] {$\li a^2$} ;
  \node(aab) [state,right of=laa] {$a^2b$} ;
  \node(abb) [state,right of=aab] {$a b^2$} ;
  \node(bbb) [state,right of=abb] {$b^3$} ;
  \node(bba) [state,right of=bbb] {$b^2 a$} ;
  \node(baa) [state,right of=bba] {$b a^2$} ;
  \node(rrr) [state,accepting,right of=baa] {$\re^3$} ;
  \node(aaa) [state,above of=bbb] {$a^3$} ;
  
  \path[->] (lll) edge node {$a$} (lla);
  \path[->] (lla) edge node {$a$} (laa);
  \path[->] (laa) edge node {$b$} (aab);
  \path[->] (aab) edge node {$b$} (abb);
  \path[->] (abb) edge node {$b$} (bbb);
  \path[->] (bbb) edge node {$a$} (bba);
  \path[->] (bba) edge node {$a$} (baa);
  \path[->] (baa) edge node {$\re^3$} (rrr);

  \path[->,dashed] (baa) edge [swap,bend right] node {$a$} (aaa);
  \path[->,dashed] (aaa) edge [swap,bend right] node {$b$} (aab);
  \path[->,dotted] (baa) edge [bend left] node {$b$} (aab);  
\end{tikzpicture}
\caption{Algorithm $\TRFC_4$ on input $\{a^3\to a^2b^3a^2\}$}
\label{ex3:fig}
\end{figure}
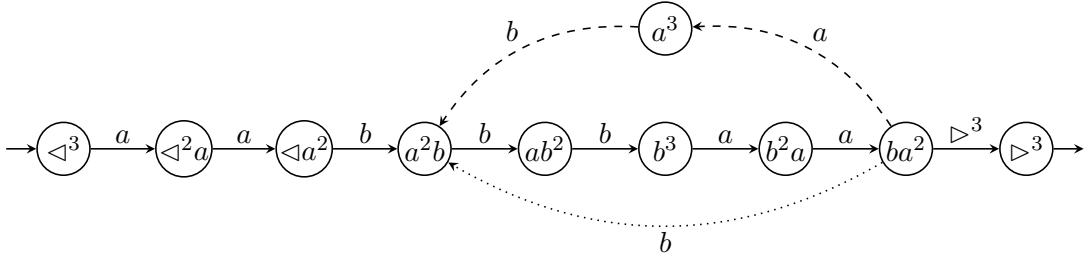

For rule $a\re^3\to a^2b^3a^2\re^3\in\Forw_4(R)$,
there is a redex starting at $b^2a$.
The corresponding reduct path needs two new edges (dashed).
The very same rule has another redex starting at $b^3$,
which needs another edge (dotted).
There is a $R$-redex from $b^3$ to $a^3$,
with the reduct path alread present, but ending in $ba^2$.
The path will be right-context-closed by $b^3$, the only continuation from $a^3$.
In fact, if we extend by just one letter,
the extended paths meet at $a^2b$ already.
This means that we do not need to consider longer extensions.
The automaton is now closed
with respect to $\CC_4(R)\cup\Forw_4(R)$.

We compute $\Btiled_T(R)$: there is just one $R$-redex, from $b^3$ to $a^3$,
and exactly one $\CC_4(R)$-redex, for $a^3b^3 \to a^2b^3a^2b^3$,
from state $b^3$ to itself.
This gives just one tiled rule
\[ [bbba, bbaa, baaa, aaab, aabb, abbb]
  \to [bbba, bbaa, baab, aabb, abbb, bbba, bbaa, baab, aabb, abbb].
\]
The letter (tile) $baaa$ does appear in the left-hand side,
but not  in the right-hand side,
thus $\Btiled_T(R)$ is terminating, implying termination of $R$.
\eex  
\end{example}

Similar to
\emph{semantic unlabeling}~\cite{DBLP:conf/rta/SternagelT11},
we can sometimes use the partial algebra for removing rules without labelling.
A similar procedure was suggested in~\cite{HW10} for rule removal
in the context of RFC matchbounds. 

\begin{algorithm}\label{imp:alg:remove}[Tiling for RFC with Untiling,
  abbreviation $\TRFCU$]
  \begin{itemize}
  \item Specification:
    \begin{itemize}
    \item Input: A rewrite system $R$ over $\Sigma$, a number $k$
    \item Output: A rewrite system $R_1\subseteq R$
      such that $\SN(R)\iff\SN(R_1)$.
    \end{itemize}
  \item Implementation:
    do the first step of Algorithm~\ref{imp:alg},
    to obtain $\aalg =(T,\Value{\cdot})$.
    Output the set of all rules $(\ell\to r)\in R$
    with $\lab_\aalg(\ell\to r)\neq \emptyset$.
  \end{itemize}
\end{algorithm}
\begin{proof} Correctness:
  As before, $\SN(R,\Lang(T))$ if and only if $\SN(R)$.
  By construction, $R$-derivations from $\Lang(T)$
  can only use rules from $R_1$.
\end{proof}

\begin{example}\label{ex:z018}
  We apply Algorithm~\ref{imp:alg:remove}, for $k=2$,
  to $R=\{ab\to bca,ba\to acb,bc\to cbb\}$.
  This is \verb|SRS/Zantema/z018| from TPDB.
  We construct the 2-shift automaton, see Figure~\ref{z018:fig},
  and we find that $\Btiled_T(ab\to bca)=\emptyset$.
  \begin{figure}[ht!]
    \[ \begin{tikzpicture}[auto, on grid, node distance=20mm, 
  inner sep=2pt, semithick, >=stealth,  
  every state/.style={minimum size=20pt, inner sep=0pt, 
  initial distance=4mm, accepting distance=4mm}, 
  initial/.style=initial by arrow, initial text=, 
  accepting/.style=accepting by arrow] 
  \node (li) [state, initial] {$\li$};
  \node (c) [state, right of=li] {$c$};
  \node (a) [state, above of=c] {$a$};
  \node (b) [state, right of=c] {$b$};
  \node (re) [state, accepting, right of=a] {$\re$};
  \path[->] (li) edge node {$a$} (a);
  \path[->] (li) edge [swap] node {$c$} (c);
  \path[->] (li) edge [swap, bend right=60] node {$b$} (b);
  \path[->] (a) edge node {$\re$} (re);
  \path[->] (c) edge [swap, bend right=8] node {$b$} (b);
  \path[->] (b) edge [swap, bend right=8] node {$c$} (c);    
  \path[->] (b) edge [swap] node {$\re$} (re);
  \path[->] (b) edge [swap] node {$a$} (a);
  \path[->] (c) edge [swap, bend right=10] node {$a$} (a);
  \path[->] (a) edge [swap, bend right=10] node {$c$} (c);
  \path[->] (b) edge [loop below] node {$b$} (b);
  \path[->] (c) edge [loop below] node {$c$} (c);
\end{tikzpicture} \]
    \caption{Algorithm $\TRFCU_2$ on input z018}
    \label{z018:fig}
  \end{figure}
  The algorithm outputs $\{ba\to acb,bc\to cbb\}$.
  Note that the automaton contains redexes
  for $(a\re\to bca\re) \in\Forw_1(ab\to bca)$  (from states $\li, c$, and $b$)
  but the criterion is the occurrence of $ab\to bca$ only.
  To handle the resulting termination problem,
  we reverse all strings in all (remaining) rules,
  obtaining $\{ab\to bca, cb\to bbc\}$.
  Again we apply Algorithm~\ref{imp:alg:remove}
  and this time we find that $ab$ does not occur in the automaton.
  This leaves $\{cb\to bbc\}$.
  Applying the algorithm one more time, we find that there is no $cb$
  in the 2-shift automaton for $\RFC(cb\to bbc)$.
  The algorithm outputs $\emptyset$, and we have proved termination
  of  \verb|z018|.
  \eex
\end{example}

In later examples, will abbreviate proof steps:
we write $R \TileAllRFC{k} S$ if Algorithm~\ref{imp:alg} transforms $R$ to $S$,
and similarly $R \TileRemoveRFC{k} S$ for Algorithm~\ref{rel-imp:alg}.
We write $R\Mirror S$ if $S$ is obtained by reversing all left-hand sides
and all right-hand sides of $R$.
The termination proof of Example~\ref{ex:z018} then reads
  \[
    R \TileRemoveRFC{2} \{ba\to acb,bc\to cbb\}\Mirror \{ab\to bca,cb\to bbc\}
    \TileRemoveRFC{2} \{cb\to bbc\}\TileRemoveRFC{2}\emptyset,
\]
and we further compress this to
\[
    (3,2) \TileRemoveRFC{2} (2,2)\Mirror(2,2)
    \TileRemoveRFC{2}(1,2)\TileRemoveRFC{2}(0,0),
  \]
  where $(r,s)$ denotes a rewrite system with $r$ rules over $s$ letters.

\begin{example}\label{z001:ex}
  We prove termination of Zantema's problem $\{a^2b^2\to b^3a^3\}$,
  a classical benchmark, in several ways.

  We only give proof outlines here, full proofs are available at
\url{https://gitlab.imn.htwk-leipzig.de/waldmann/pure-matchbox/tree/master/sparse-tiling-data}.
  In the proof outline, $\Weight$ denotes removal of rules by weights.

  There is a proof with tiles of width 2 only,
  using several steps:
\begin{align*}
  (1,2)\TileAllRFC{2}(4,4)\TileAllRFC{2}(16,8)\TileAllRFC{2}(49,15)\TileAllRFC{2}(121,26)\Weight(64,26)\TileRemoveRFC{2}(60,26) \\
  \TileAllRFC{2}(153,44)\Weight(105,44)\TileAllRFC{2}(312,68)\Weight(220,68)\TileRemoveRFC{2}(160,64)\TileAllRFC{2}(372,95) \\
  \TileRemoveRFC{2}(332,95)\TileAllRFC{2}(629,138)\Weight(208,138)\TileRemoveRFC{2}(42,102)\Weight(16,102) \\
  \TileAllRFC{2}(28,80)\Weight(24,80)\TileAllRFC{2}(32,94)\Weight(4,94)\TileRemoveRFC{2}(0,0)
\end{align*}

For width 12, there is a proof with just one step, but the intermediate system is large:
\begin{align*}
(1,2)\TileAllRFC{12}(1166,344)\Weight(0,344)
\end{align*}

There even is a termination proof that does not use weights at all:
\begin{align*}
  (1, 2)\TileAllRFC{2}(4, 4)\TileRemoveRFC{5}(3, 4)\TileAllRFC{3}(40, 15)\TileAllRFC{2}(105, 26)\TileRemoveRFC{5}(65, 26)\TileRemoveRFC{5}(52, 26) \\
  \TileRemoveRFC{5}(37, 26)\TileAllRFC{2}(97, 44)\TileRemoveRFC{5}(37, 43)\TileRemoveRFC{5}(36, 43)\TileAllRFC{2}(110, 68)\TileRemoveRFC{5}(80, 64) \\
  \TileAllRFC{2}(192, 93)\TileRemoveRFC{5}(96, 89)\TileRemoveRFC{3}(58, 79)\TileRemoveRFC{5}(32, 66)\TileRemoveRFC{3}(0, 0).
\end{align*}
\eex
\end{example}

\begin{example}
  We show that our method can be applied as a preprocessor
  for other termination provers.
  We consider $R= \{0000 \to 1001, 0101 \to 0010\}$, 
  which is \verb|SRS/Gebhardt/16| from the TPDB.
  After the chain of transformations
  \begin{align*}
(2, 2)\TileAllRFC{3}(98, 20)\Weight(24, 11)\TileRemoveRFC{2}(17, 10)\Weight(15, 8),
  \end{align*}
  the resulting problem can be solved by \TTTT~\cite{DBLP:conf/rta/KorpSZM09}
  quickly, via KBO.
  \TTTT\ did not solve this problem in the Termination Competition 2018.
\eex
\end{example}

\section{Relative Termination Proofs via Overlap Closures}\label{sec:rel}

We now apply our approach to prove relative termination.
With relative termination, the $\RFC$ method does not work.

\begin{example}\label{rel:ex:rfc}
  $R/S$ may nonterminate although $R/S$ terminates on $\RFC(R\cup S)$.
  For example, let $R=\{ab\to a\}$ and $S=\{c\to bc\}$.
  We have $\RFC(R\cup S)=a \cup b^+ c$.
  This does not have a factor $ab$, therefore $\SN(R/S)$ on $\RFC(R\cup S)$.
  On the other hand, $\neg\SN(R/S)$ because of the loop
  $\underline{ab}c\to_R a\underline{c}\to_S abc$.
\eex
\end{example}

Therefore, we use overlap closures instead.
To prove correctness of this approach, we use a characterization of overlap
closures as derivations in which every position between letters is touched.
A new left-recursive characterization
of overlap closures (Corollary~\ref{ct:cor:main})
allows us to enumerate right-hand sides of overlap closures by completion.

\subsection{Overlap Closures}
A position between letters in the starting string of a derivation is called
\emph{touched} by the derivation if it has no residual in the final string.

\begin{example}
  For the rewrite system $R = \{ab \to baa\}$ over alphabet $\{a,b\}$, all
  positions labelled by $|$ in the starting string $a|a|ba|b$ are touched by
  the derivation $aabab\to_R abaaab\to_R baaaaab\to_R baaaabaa$. The position
  between $b$ and $a$ in the starting string has the residual position
  between $a$ and $b$ in the final string.
\end{example}

Let $\OC(R)$ denote the set of overlap closures~\cite{DBLP:journals/siamcomp/GuttagKM83},  
and let $\ROC(R) = \rhs(\OC(R))$.

\begin{lemma}\cite[Lemma~3]{DBLP:journals/ita/GeserZ99}\label{def:rel:roc}
  The set $\OC(R)$ of overlap closures of $R$ is the set of all
  $R$-derivations where all initial positions between letters are touched.
\end{lemma}

Termination has been characterized
by forward closures~\cite{DBLP:conf/icalp/Dershowitz81}.
In the following we obtain a characterization of relative termination
by overlap closures.

\begin{definition}
  For a finite or infinite $R$-derivation $A$,
  let $\Inf(A)$ denote the set of rules that are applied 
  infinitely often in $A$. (For a finite derivation, $\Inf(A)=\emptyset$.)
\end{definition}

\begin{proposition}\label{prop:rel:split}
  For each $R$-derivation $A$, 
  there are finitely many $R$-derivations $B_1,\ldots,B_k$
  that start in $\ROC(R)$, and $\Inf(A)=\bigcup_i \Inf(B_i)$.
\end{proposition}
\begin{proof}
  If $A$ is empty, then $k=0$.
  If $A$ has a finite prefix that is an $\OC$, 
  then $k=1$ and $B_1$ is the (infinite) suffix.
  Else, the start of $A$ has a position
  that is never touched during $A$.
  We can then split the derivation, and use induction
  by the length of the start of the derivation.
\end{proof}

\begin{proposition}\label{prop:rel:cap}
  $\SN(R/S)$ if and only if
  for each $(R \cup S)$-derivation $A$, $\Inf(A)\cap R=\emptyset$.
\end{proposition}

The following theorem says that for analysis of relative termination,
we can restrict to derivations starting from right-hand sides of overlap closures.
\begin{theorem}\label{thm:rel:sn-roc}
  $\SN(R/S)$ if and only if $\SN(R/S)$ on $\ROC(R\cup S))$.
\end{theorem}
\begin{proof}
  The implication from left to right is trivial, as we consider
  a subset of derivations.
  For the other direction, let $A$ be an $(R\cup S)$-derivation. 
  Using Proposition~\ref{prop:rel:split} we obtain $B_1,\ldots, B_k$ for $A$ 
  such that 
\[ 
\Inf(A) \cap R= (\bigcup_i \Inf(B_i))\cap R = \bigcup_i (\Inf(B_i)\cap R)
  = \bigcup_i \emptyset = \emptyset,
\]
  thus $\SN(R/S)$ by Proposition~\ref{prop:rel:cap}. 
\end{proof}

\subsection{Tiling for Overlap Closures}\label{sec:rel-imp}

We employ the following left-recursive characterisation of
$\ROC(R)$ (proved in the Appendix)
that is suitable for a completion algorithm.

\begin{corollary}\label{ct:cor:main}
$\ROC(R)$ is the least set $S$ such that
\begin{enumerate}
\item \label{ct:cor:main:1} $\rhs(R)\subseteq S$,
\item \label{ct:cor:main:2} if $tx\in S$ and $(xu,v)\in R$
for some $t,x,u\neq \emptystring$ then $tv\in S$;
\item \label{ct:cor:main:3} if $xt\in S$ and $(ux,v)\in R$
for some $t,x,u\neq \emptystring$ then $vt\in S$;
\item \label{ct:cor:main:4} if $tut'\in S$ and $(u,v)\in R$
then $tvt'\in S$;
\item \label{ct:cor:main:5} if $tx\in S$ and $yv\in S$ and $(xwy,z)\in R$
for some $t,x,y,v\neq \emptystring$ then $tzv\in S$.
\end{enumerate}
\end{corollary}

Note that
Item~\ref{ct:cor:main:4} is the standard rewrite relation of $R$.
Item~\ref{ct:cor:main:2} is suffix rewriting,
and we already simulate this with $\Forw_k(R)$, see Proposition~\ref{prop:rfc}.
Item~\ref{ct:cor:main:3} is prefix rewriting,
and it can be handled symmetrically by left end markers
\[ \Backw_k(R) = \{\li^{k-1}\ell_2\to \li^{k-1}r\mid (\ell_1\ell_2\to r)\in R, \ell_1\neq\epsilon\neq \ell_2\} .\]

Item~\ref{ct:cor:main:5}  is an inference rule with two premises
that cannot be written as a rewrite relation.
We can still apply the tiling method, with the following modification.

Premise $(xwy,z)\in R$ refers to some suffix $x$ of $S$,
and some prefix $y$ of $S$, and to some unspecified $w$.
We aim to represent such $xwy$ by a path in the automaton.
Starting from the automaton constructed in Section~\ref{sec:trs},
we add a path from final state $\re^{k-1}$
to initial state $\li^{k-1}$, consisting of $k-1$ transitions labelled $\li$.
Note that this is still a shift automaton.

Then an application of Item~\ref{ct:cor:main:5}
of Corollary~\ref{ct:cor:main}
with $(xwy,t)\in R$
is realized by a standard rewrite step
$x\re^{k-1}\li^{k-1}y \to t$.
The extra path is used to trace $w$.

Similar to Definition~\ref{trs:def:pa},
Proposition~\ref{trs:prop:pa},
Definition~\ref{trs:def:cc},
we have
\begin{definition}\label{rel-imp:def:pa}
  For $T\subseteq \Btiles_k(\Sigma^*)$,
  the \emph{looped} shift algebra $\PAL_k(T)$
  is $\PA_k(T\cup \Tiles_k(\re^{k-1}\li^{k-1})),$
  over signature $\Sigma\cup\{\epsilon,\li,\re\}.$
\end{definition}

By construction, $\PA_k(T)$ is a sub-algebra of $\PAL_k(T)$.
The language of $\PAL_k(T)$ is the prefix closure of
$\Lang(T)\re^{k-1}\left(\li^{k-1} \Lang(T)\re^{k-1}\right)^*.$

\begin{definition}
  Let $\CCL_k(R)=
  \{ (z)x\re^{k-1}\li^{k-1} ye\to(z)re
  \mid (xwy\to r)\in R, x\neq \epsilon\neq y, e\in\Tiles_{k-1}(\Sigma^*\re^*)
  \}$.
\end{definition}
The purpose of this construction is:
\begin{proposition}
  For a set of $k$-tiles $T$ and a rewrite system $R$,
  if $\Lang_{\PAL_k(T)}$ is closed with respect to
  $\CC_k(R)\cup\Forw_k(R)\cup\Backw_k(R)\cup\CCL_k(R)$,
  then $\ROC(R)\re^{k-1}\subseteq \Lang_{\PAL_k(T)}$.
\end{proposition}

We transform global relative termination as follows:
\begin{algorithm}\label{rel-imp:alg}[Tiling for ROC, abbreviation $\TROC$]
  \begin{itemize}
  \item Specification:
    \begin{itemize}
    \item Input: Rewrite systems $R_1,R_2$ over $\Sigma$, number $k$
    \item Output: Rewrite systems $R_1',R_2'$ over $\Btiled_k(\Sigma)$
      such that $\SN(R_1/R_2)\iff\SN(R_1'/R_2')$.
    \end{itemize}
  \item Implementation:
    Let $R=R_1\cup R_2$.
    We call Algorithm~\ref{trs:algo:algebra} with
    \begin{enumerate}
    \item the term rewrite system
      $\CC_k(R)\cup\Forw_k(R)\cup\Backw_k(R)\cup\CCL_k(R)$,
    \item the looping $k$-shift algebra over signature $\Sigma$
    \item and the set of domain elements $\Prefix_{k-1}(T)\cup\Suffix_{k-1}(T)$,
      for  $T=\Btiles_k(\rhs(R))\cup \Tiles_k(\re^*\li^*)$.
    \end{enumerate}
    We obtain a partial algebra $\aalg=(A,\Value{\cdot})$.
    We output $(\Btiled_U(R_1),\Btiled_U(R_2))$, where $U=\Tiles(\aalg)$.
\end{itemize}
\end{algorithm}
\begin{proof} Correctness:
  $\Lang(T)$ contains $\rhs(R_1\cup R_2)$.
  
  $\PAL_k$ is a model for $\CC_k(R)\cup\Forw_k(R)\cup\Backw_k(R)\cup\CCL_k(R)$
  by construction: all rules keep the suffix of length $k-1$  intact.

  The precondition of Algorithm~\ref{trs:algo:algebra} is satisfied,
  so we get  $\aalg$ as a partial model
  for  $\CC_k(R)\cup\Forw_k(R)\cup\Backw_k(R)\cup\CCL_k(R)$
  that contains $\Lang(T)\re^{k-1}$, and thus, $\rhs(R)\re^{k-1}$.

  By Theorem~\ref{trs:thm:rel}, $\SN(R_1'/R_2')$ iff $\SN(R_1/R_2,\Lang(U))$.
  By $\ROC(R)\subseteq \Lang(U)$, this implies $\SN(R_1/R_2,\ROC(R))$.
  By Theorem~\ref{thm:rel:sn-roc}, this implies $\SN(R_1/R_2)$.
  For the other direction,  $\SN(R_1/R_2)$ implies  $\SN(R_1/R_2,L)$
  for any language $L$, in particular, for $L=\Lang(U)$.
\end{proof}

Note that the first step of this algorithm
makes no distinction between strict rules ($R_1$) and weak rules ($R_2$):
the algebra  $\aalg$ is constructed starting from $\rhs(R_1\cup R_2)$,
and closed with respect to $R_1\cup R_2$.

\begin{example}
  We illustrate Algorithm~\ref{rel-imp:alg}
  for $R_1=\{a^3\to a^2b^2a^2\}$ and $R_2=\emptyset$.
  We take $k=4$ and start with the automaton for $\rhs(R)$,
  and include the backwards path from $\re^3$ to $\li^3$
  (the solid arrows in Figure~\ref{ex6:fig}).
  
\begin{figure}[ht!]
  \begin{tikzpicture}[auto, on grid, node distance=16mm, 
  inner sep=2pt, semithick, >=stealth,  
  every state/.style={minimum size=20pt, inner sep=0pt, 
  initial distance=4mm, accepting distance=4mm}, 
  initial/.style=initial by arrow, initial text=, 
  accepting/.style=accepting by arrow]
  
  \node(lll) [state,initial] {$\li^3$} ;
  \node(lla) [state,right of=lll] {$\li^2a$} ;
  \node(laa) [state,right of=lla] {$\li a^2$} ;
  \node(aab) [state,right of=laa] {$a^2b$} ;
  \node(abb) [state,right of=aab] {$a b^2$} ;
  \node(bba) [state,right of=abb] {$b^2 a$} ;
  \node(baa) [state,right of=bba] {$b a^2$} ;
  \node(rrr) [state,accepting,right of=baa] {$\re^3$} ;
  \node(aaa) [state,above of=bba] {$a^3$} ;
  
  \path[->] (lll) edge node {$a$} (lla);
  \path[->] (lla) edge node {$a$} (laa);
  \path[->] (laa) edge node {$b$} (aab);
  \path[->] (aab) edge node {$b$} (abb);
  \path[->] (abb) edge node {$a$} (bba);
  \path[->] (bba) edge node {$a$} (baa);
  \path[->] (baa) edge node {$\re^3$} (rrr);

  \path[->] (rrr) edge [bend left, swap] node {$\li^3$} (lll);
  \path[->,dashed] (baa) edge [swap,bend right] node {$a$} (aaa);
  \path[->,dashed] (aaa) edge [swap,bend right] node {$b$} (aab);
  \path[->,dotted] (baa) edge [bend left] node {$b$} (aab);  
\end{tikzpicture}
\caption{Algorithm $\TROC_4$ on input $\{a^3\to a^2b^2a^2\}$}
\label{ex6:fig}
\end{figure}
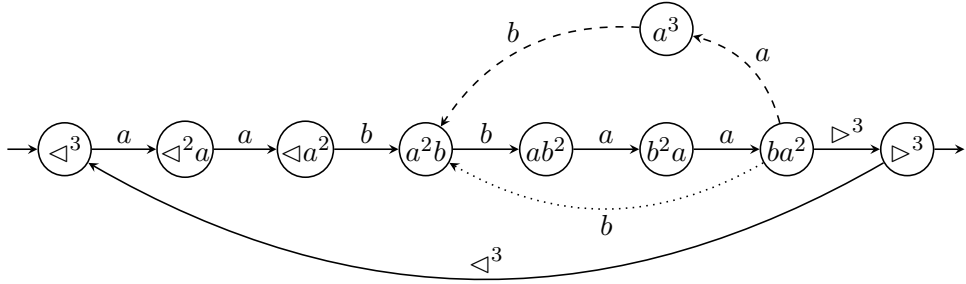
We now consider rules $(a\re^3\li^3 ae \to a^2b^2a^2e)\in \CCL(R)$.
These can only start at state $b^2a$, 
and the only choice for the right 3-context $e$ in those rules is $abb$.
The reduct path needs two fresh edges (dashed).
For rules $(a^2\re^3\li^3 ae \to a^2b^2a^2e)\in \CCL(R)$,
a redex must start in $ab^2$, and the only right 3-context $e$ is still $abb$.
The reduct path needs one extra edge (dotted).
The automaton is now closed also with respect to the other operations.%
We compute $\Btiled_T(R_1)$. There is just one $R_1$-redex, starting at $ab^2$,
with just one right extension $bba$. This creates just one labelled rule
\[
  [ abba, bbaa, baaa, aaab, aabb, abba]
  \to [abba, bbaa, baab, aabb, abba, bbaa, baab, aabb, abba].
\]
\eex
\end{example}

It is often the case that $\SN(\Btiled_T(R_1)/\Btiled_T(R_2))$
can be obtained with some easy method, e.~g., weights.

\begin{example}\label{rbeans:ex}
The \emph{bowls and beans} problem had been suggested
by Vincent van Oostrom~\cite{bowls-and-beans}.
It asks to prove termination of this relation:
\begin{quote}
\emph{If a bowl contains two or more beans, pick any two beans in it and move
one of them to the bowl on its left and the other to the bowl on its right.}
\end{quote}
In a direct model, a configuration is a function $\ZZ\to\NN$ with
finite support.
In a rewriting model, this is encoded as a string.
Several such models have been submitted to TPDB
by Hans Zantema (\verb|SRS_Standard/Zantema_06/beans[1..7]|).
We consider here a formalisation as a relative termination problem
(\verb|SRS_Relative/Waldmann_06_relative/rbeans|).
\[ \{baa \to abc, ca\to ac, cb\to ba\}/\{\epsilon\to b\} \]
Here, $a$ is a bean, $b$ separates adjacent bowls,
and $c$ transports a bean to the next bowl.
The relative rule is used to add extra bowls at either end ---
although it can be applied anywhere, meaning that any bowl
can be split in two, anytime, which does not hurt termination.
To the best of our knowledge, this benchmark problem
had never been solved in a termination competition.

We can now give a termination proof via tiling of width 3,
and using overlap closures:
\begin{align*}
(3/1,3)\TileAllROC{3}(416/144,47)\Weight(207/48,34)\Matrix{\Natural}{2}(63/48,32)\Weight(0/33,22)
\end{align*}
Here, notation $(r/s,a)$ stands for a relative
termination problem $\SN(R/S)$ where $R$ has $r$ rules,
$S$ has $s$ rules, and the alphabet has $a$ letters;
and $\TileAllROC{k}$ %
denote an application of Algorithm~\ref{rel-imp:alg},
and $\Matrix{\Natural}{d}$ denotes rule removal
by matrix interpretation of natural numbers with dimension $d$.
\eex
\end{example}

Similar to Algorithm~\ref{imp:alg:remove},
there is a variant that uses tiling to return a subset of rules.

\begin{algorithm}\label{rel-imp:alg:remove}[Tiling for ROC with Untiling,
  abbreviation $\TROCU$]
  \begin{itemize}
  \item Specification:
    \begin{itemize}
    \item Input: Rewrite systems $R,S$ over $\Sigma$, a number $k$
    \item Output: Rewrite systems $R_1\subseteq R$, $S_1\subseteq S$
      such that $\SN(R/S)\iff\SN(R_1/S_1)$.
    \end{itemize}
  \item Implementation:
    Apply the first step of Algorithm~\ref{rel-imp:alg},
    to obtain $\aalg =(A,\Value{\cdot})$.
    Output $R_1= \{ (\ell\to r) \mid (\ell\to r)\in R, \emptyset\neq \lab_\aalg(\ell\to r)\}$
    and  $S_1= \{ (\ell\to r) \mid (\ell\to r)\in S, \emptyset\neq \lab_\aalg(\ell\to r)\}$.
  \end{itemize}
\end{algorithm}
\begin{proof} Correctness:
  As before, $\SN(R/S,\Lang(T))\iff \SN(R/S)$.
  By construction, $(R\cup S)$-derivations from $\Lang(T)$
  can only use rules of $(R_1\cup S_1)$.
\end{proof}

\begin{example}\label{r4:ex}
  $\SN(ababa \to \epsilon / ab\to bbaa)$
  (\verb|SRS_Relative/Waldmann_06_relative/r4| from TPDB)
  can be solved quickly by $\TROCU(4)$.
  The set $T$ of tiles has 28 elements,
  and $baba\notin T$.
  This implies that $\Btiled_T(ababa\to\epsilon)$ is empty,
  so $\TROCU(4,ababa \to \epsilon / ab\to bbaa)$
  is $(\emptyset/ab\to bbaa)$, for which $\SN$ holds trivially.

  In the Termination Competition 2018,
  AProVE~\cite{DBLP:journals/jar/GieslABEFFHOPSS17}
  solved this benchmark with double root labelling,
  which is very similar to tiling of width 3,
  but this took more than 4 minutes.
\eex
\end{example}

The following example applies Algorithm~\ref{rel-imp:alg}
to a relative termination problem that comes from
the dependency pairs transformation~\cite{DBLP:journals/tcs/ArtsG00}.

\begin{example}\label{w16:ex}
  The system 
  $\{ abababaababa\to ababaababaabab \}$
  is part of the enumeration \verb|SRS_Standard/Wenzel_16|,  
  and it was not solved in the Termination Competition 2018.
  We obtain a termination proof with outline
  \begin{align*}
    (1,2)\Deepee(8/1,3)\TileAllROC{2}(16/8,8)\Weight(12/8,8)\\
    \TileRemoveROC{4}(12/4,6)\TileAllROC{8}(162/1782,185)\Weight(0/420,185).
  \end{align*}
  Here, $\Deepee$ stands for the dependency pairs transformation.
  There is a shorter proof with larger tiles
  \begin{align*}
    (1,2)\Deepee(8/1,3)\TileAllROC{11}(208/3952,290)\Weight(0/912,290).
  \end{align*}
  There are two more systems $\{ababaababa\to abaabababaab\}$,
  $\{abaababaab\to aababaabaabab\}$ of \verb|SRS_Standard/Wenzel_16|,
  that were unsolved in then Termination Competition 2018,
  and can now proved terminating automatically via Algorithms $\TROC$
  and $\TROCU$. Intermediate systems have up to 3940 rules.
\eex
\end{example}

\section{Experimental Evaluation}\label{sec:experiments}

Sparse tiling is implemented in the termination prover
Matchbox\footnote{\url{https://gitlab.imn.htwk-leipzig.de/waldmann/pure-matchbox}}
that won the categories \emph{SRS Standard} and \emph{SRS Relative}
in the Termination Competition 2019.
Matchbox employs a parallel proof search with a portfolio of algorithms,
including Algorithm~\ref{rel-imp:alg}. 

For relative termination, we use weights, matrix interpretations
  over the naturals,
  and tiling of widths 2, 3, 5, 8 (in parallel), cf.\ Example~\ref{rbeans:ex}.
For standard termination, we use RFC matchbounds,
  and (in parallel) the dependency pairs (DP) transformation,
  creating a relative termination problem,
  to which we apply weights, matrix interpretations
  over natural and arctic numbers, and tiling of width 3 (only).

Table~\ref{experiments:table:yes} shows performance of variants of these strategies
on SRS benchmarks of TPDB,
as measured on Starexec, under the Termination profile
(5 minutes wall clock, 20 minutes CPU clock, 128 GByte memory).
In all experiments,
we keep using weights and (for standard termination) the DP transform.
The bottom right entry of each sub-table contains
the result for the full strategy, used in competition.
\begin{table}[ht!]
  \begin{center}
  \begin{tabular}{rr|r|r}
    \multicolumn{2}{l|}{\text{SRS Relative}} & \multicolumn{2}{c}{\text{matrices}} \\ 
    \multicolumn{2}{l|}{Starexec Job 33975} & \text{no} & \text{yes} \\ \hline
    \multirow{2}{*}{tilling} & \text{no}  & 1 & 72 \\ \cline{2-4}
    & yes & 176 & 225
  \end{tabular}
  \quad
  \begin{tabular}{rr|r|r}
    \multicolumn{2}{l|}{\text{SRS Standard}} & \multicolumn{2}{c}{\text{RFC matchbounds, matrices}} \\ 
   \multicolumn{2}{l|}{Starexec Job 33976} & \text{none} & \text{both} \\ \hline
    \multirow{2}{*}{tilling} & \text{no}  & 100 & 1122 \\ \cline{2-4}
    & yes & 512 & 1133
  \end{tabular}
\end{center}
\caption{Number of termination proofs obtained by variants of Matchbox}
\label{experiments:table:yes}
\end{table}

We note a strong increase in the last column (matrices:yes) of the left sub-table.
We conclude that
sparse tiling is important for relative termination proofs.
The right sub-table shows a very weak increase in the corresponding column.
We conclude that with Matchbox' current search strategy
for standard termination, other methods overshadow tiling,
e.~g., RFC matchbounds are used in 578 proofs, and arctic matrices in 389 proofs.

For relative termination, 
the method of tiling, with weights, but without matrices,
is already quite powerful with 176 proofs,
a number between those for AProVE (163) and MultumNonMulta (192).

Table~\ref{experiments:table:widths}
shows the widths used in tiling proofs for relative SRS.
The sum of the bottom row is greater than the total number of proofs (225)
since one proof may use several widths.
\begin{table}[ht!]
\begin{center}
  \begin{tabular}{p{3cm}|r|r|r|r}
    width  & 2 & 3 & 5 & 8 \\ \hline
    proofs & 150 & 57 & 38 & 11
  \end{tabular}
\end{center}
\caption{Number of termination proofs for relative SRS, using given width of tiling}
  \label{experiments:table:widths}
\end{table}

We observe that short tiles appear more often.
We think the reason is that larger tiles
tend to create larger systems that are more costly to handle,
while resources (time and space on Starexec) are fixed.
This is also the reason for using width 3 only, for standard termination.

\section{Conclusion}\label{sec:conc}

We have presented \emph{sparse tiling},
a method to compute a regular over-approximation
of reachability sets, using sets of tiles, represented as automata,
and we applied this to the analysis of termination and relative termination.
The method is an instance of semantic labelling via a partial algebra.
Our contribution is the choice of the $k$-shift algebra.

We also provide a powerful implementation in Matchbox
that contributed to winning
the SRS categories in the Termination Competition 2019.
An exact measurement of that contribution is difficult
since termination proof search (in Matchbox)
depends on too many parameters.

Interesting open questions (that are independent of any implementation)
are about the relation between sparse tilings of different widths,
and between sparse tilings and other methods, e.~g., matchbounds.
Since our focus for the present paper is string rewriting,
we also leave open the question of whether sparse tiling would be
useful for termination of term rewriting.

\bibliographystyle{alpha}
\newcommand{\etalchar}[1]{$^{#1}$}

\appendix
\section{Composition Trees of Overlap Closures}

In this section we derive a left-recursive characterization of overlap closures
in string rewriting. By left-recursive, we mean that the recursive descent
takes place only in the left partners.
The definition of overlap closures recurses in both arguments as 
we always overlap a closure with a closure.
(Note: Since $R$ is fixed throughout this section, we simply write $\OC$ instead
of $\OC(R)$, analogously for other operators.)

\begin{definition}\label{app:def:oc}\cite{DBLP:journals/siamcomp/GuttagKM83}
For a rewrite system $R$, the set $\OC$ is defined as the least set such that
\begin{enumerate}
\item[(1)] $R\subseteq \OC$,
\item[(2)] if $(s,tx)\in \OC$ and $(xu,v)\in \OC$
  for some $t,x,u\neq \emptystring$ then $(su,tv)\in \OC$;
\item[(2')] if $(s,xt)\in \OC$ and $(ux,v)\in \OC$
  for some $t,x,u\neq \emptystring$ then $(us,vt)\in \OC$;
\item[(3)] if $(s,tut')\in \OC$ and $(u,v)\in \OC$
  then $(s,tvt')\in \OC$;
\item[(3')]
  if $(u,v)\in\OC$ and $(svs',t)\in \OC$
  then $(sus',t)\in \OC$.
\end{enumerate}
\end{definition}

The following recursive definition is left-recursive
(we overlap a closure with a rule).
We need an extra rule (Item~4)
and drop a rule (Item~3'),
the others correspond to Definition~\ref{app:def:oc}.
\begin{definition}\label{app:def:oc'}
For a rewrite system $R$, the set $\OC'$ is defined as the least set such that
\begin{enumerate}
\item[(1)] $R\subseteq \OC'$,
\item[(2)] if $(s,tx)\in \OC'$ and $(xu,v)\in R$
for some $t,x,u\neq \emptystring$ then $(su,tv)\in \OC'$;
\item[(2')] if $(s,xt)\in \OC'$ and $(ux,v)\in R$
for some $t,x,u\neq \emptystring$ then $(us,vt)\in \OC'$;
\item[(3)] if $(s,tut')\in \OC'$ and $(u,v)\in R$
then $(s,tvt')\in \OC'$;
\item[(4)]
  if $(s,tx)\in \OC'$ and $(u,yv)\in \OC'$ and $(xwy,z)\in R$
for some $t,x,y,v\neq \emptystring$ then $(swu,tzv)\in \OC'$.
\end{enumerate}
\end{definition}

The main result of this Appendix is
that the set $\OC'$ covers the overlap closures up to inverse rewriting
of left hand sides.
Let $\OC_N := \{(s,t) \mid  s\to_R^* s' \land (s',t)\in \OC'\}$.
\begin{theorem}\label{thm:oc}
$\OC = \OC_N$.
\end{theorem}
Since we are interested in right-hand sides of closures,
the extra rewrite steps in $\OC_N$ do not hurt.

In order to prove Theorem~\ref{thm:oc},
it is useful to represent a closure by a tree that describes the way
the closure is formed: the composition tree of the closure.
Each node of a composition tree denotes an application
of one of the inference rules of
Definitions~\ref{app:def:oc} and~\ref{app:def:oc'}.
An extra node type 3' denotes
an $\to_R$-step as seen in Theorem~\ref{thm:oc}.

\newcommand{\CT}{\textsf{CT}}

\begin{definition}[\cite{DBLP:journals/ita/GeserZ99}]
  Define the signature
  $\Omega = \{1, 2, 2', 3, 3', 4\}$,
  where $1$ is unary, $4$ is ternary, and the other symbols are binary. 
  The set $\CT$ of \emph{composition trees} is defined as the set of ground
  terms over $\Omega$.
\end{definition}

\begin{definition}
A composition tree represents a set of string pairs, as follows:
\begin{align*}
\sem{1} &= \{(\ell,r) \mid (\ell\to r)\in R\},\\
\sem{2(c_1,c_2)} &=
  \{(su,tv) \mid (s,tx)\in \sem{c_1}, (xu,v)\in \sem{c_2},
  t,x,u\neq \emptystring\},\\
\sem{2'(c_1,c_2)} &=
  \{(us,vt) \mid (s,xt)\in \sem{c_1}, (ux,v)\in \sem{c_2},
  t,x,u\neq \emptystring\},\\
\sem{3(c_1,c_2)} &=
  \{(s,tvt') \mid (s,tut')\in \sem{c_1}, (u,v)\in \sem{c_2}\},\\
\sem{3'(c_1,c_2)} &=
  \{(sus',t) \mid (svs',t)\in \sem{c_1}, (u,v)\in \sem{c_2}\},\\
\sem{4(c_1,c_2,c_3)} &=
  \{(swu,tzv) \mid (s,tx)\in \sem{c_1}, (u,yv)\in \sem{c_2},\\
&\quad \phantom{\{(swu,tzv) \mid } \quad
  (xwy,z)\in \sem{c_3}, t,x,y,v\neq \emptystring\} \enspace .
\end{align*}
This is conveniently extended to sets $S$ of composition trees:
\[
\sem{S} = \bigcup_{c\in S} \sem{c}.
\]
\end{definition}

\begin{example}
  The composition tree $4(1,2(1,1),3'(1,1))$ denotes all pairs
  obtained by the following overlaps of rewrite steps.
  Times flows from top to bottom. Each of the rectangles
  of height 1 is a step,
  corresponding to a 1 node in the tree.
  The grey rectangle in the top right is $2(1,1)$,
  the grey rectangle in the bottom is $3'(1,1)$.
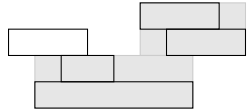
\[
  \begin{tikzpicture}[scale=.35]
    \filldraw (5,2) [fill=black!10,draw=black!20] rectangle (9,4);
    \filldraw (1,0) [fill=black!10,draw=black!20] rectangle (7,2);
    \draw (0,2) rectangle (3,3);
    \draw (1,0) rectangle (7,1);
    \draw (2,1) rectangle (4,2);
    \draw (5,3) rectangle (8,4);
    \draw (6,2) rectangle (9,3);
  \end{tikzpicture}
  \]
\end{example}

Let $\CT_0$ denote the composition trees that do not contain
the function symbol~$4$. By construction we have:
\begin{lemma}
$\OC = \sem{\CT_0}$.
\end{lemma}

Adding symbols~$4$ does not increase expressiveness,
since $\sem{4(c_1,c_2,c_3)} \subseteq \sem{2(c_1,2'(c_2,c_3))}$.
\begin{lemma}\label{lem:oc-ct}
$\OC = \sem{\CT}$.
\end{lemma}

In the remainder of this section,
we give a semantics-preserving transformation
from $\CT$ (arbitrary composition trees)
to a subset that describes the right-hand side of Theorem~\ref{thm:oc}.
Let us first characterize the goal precisely.
\begin{definition}
  The set $\CT_N$ is given by the regular tree grammar
  with variables $T,D$ (top, deep), start variable $T$, and rules
\begin{align*}
  T  \to 3'(1,T) \mid D, %
  \qquad
  D  \to 1 \mid 2(D,1) \mid 2'(D,1) \mid 3(D,1) \mid 4(D,D,1).
\end{align*}
\end{definition}

  Rules for $D$ correspond to the rules of Definition~\ref{app:def:oc'},
  creating $(s',t)\in\OC'$. Rules for $T$ correspond to the
  initial derivation $s\to_R^* s'$. Therefore,
\begin{lemma}\label{lem:oc-nf-neu}
  $\sem{\CT_N} = \OC_N$.
\end{lemma}

We are going to construct a term rewrite system $Q$ on $\Omega$
that has $\CT_N$ as its set of normal forms. It must remove
all non-1 symbols from the left argument of $3'$,
and remove all non-1 symbols from the rightmost argument of $2,2',3,$ and $4$.
Also, it must remove all $3'$ that are below some non-$3'$.
These conditions already determine the set of left-hand sides of $Q$.

For each left-hand side $\ell$, the set of right-hand sides
must cover $\ell$ semantically:
\[
  \forall \ell\in\lhs(Q): \sem{\ell}\subseteq \bigcup_{(\ell,r)\in Q}\sem{r}.
\]

A term rewrite system $Q$ over signature $\Omega$
with the desired properties is defined
in Table~\ref{tab:comp-trs}.
We bubble-up 3' symbols,
e.~g., $2(3'(c_1,c_2),c_3)\to 3'(c_1,2(c_2,c_3))$ (Rule~\ref{rewr:23'*}),
and we rotate to move non-1 symbols,
e.~g.,  $2(c_1,2(c_2,c_3))\to 2(2(c_1,c_2),c_3)$ (Rule~\ref{rewr:22a}).
Rotation below $3'$ goes from left to right,
all other rotations go from right to left.
Rules~\ref{rewr:22'a} and~\ref{rewr:2'2a} show that symbol $4$
cannot be avoided.

\begin{table}[p]
\begin{minipage}[t]{68mm}
{\footnotesize
\begin{align}
2(c_1,2(c_2,c_3)) &\to 2(2(c_1,c_2),c_3)\label{rewr:22a}\\
2(c_1,2(c_2,c_3)) &\to 2(3(c_1,c_2),c_3)\label{rewr:22b}\\
2(c_1,2'(c_2,c_3)) &\to 4(c_1,c_2,c_3)\label{rewr:22'a}\\
2(c_1,2'(c_2,c_3)) &\to 3(2(c_1,c_2),c_3)\label{rewr:22'b}\\
2(c_1,3(c_2,c_3)) &\to 3(2(c_1,c_2),c_3)\label{rewr:23}\\
2(c_1,3'(c_2,c_3)) &\to 3'(c_2,2(c_1,c_3))\label{rewr:23'a}\\
2(c_1,3'(c_2,c_3)) &\to 2(2(c_1,c_2),c_3))\label{rewr:23'b}\\
2(c_1,3'(c_2,c_3)) &\to 2(3(c_1,c_2),c_3))\label{rewr:23'c}\\
2(3'(c_1,c_2),c_3) &\to 3'(c_1,2(c_2,c_3))\label{rewr:23'*}\\
2(c_1,4(c_2,c_2',c_3)) &\to 4(2(c_1,c_2),c_2',c_3)\label{rewr:24a}\\
2(c_1,4(c_2,c_2',c_3)) &\to 4(3(c_1,c_2),c_2',c_3)\label{rewr:24b}\\
2(c_1,4(c_2,c_2',c_3)) &\to 3(3(2(c_1,c_2'),c_2),c_3)\label{rewr:24c}\\
2'(c_1,2(c_2,c_3)) &\to 4(c_1,c_2,c_3)\label{rewr:2'2a}\\
2'(c_1,2(c_2,c_3)) &\to 3(2'(c_1,c_2),c_3)\label{rewr:2'2b}\\
2'(c_1,2'(c_2,c_3)) &\to 2'(2'(c_1,c_2),c_3)\label{rewr:2'2'a}\\
2'(c_1,2'(c_2,c_3)) &\to 2'(3(c_1,c_2),c_3)\label{rewr:2'2'b}\\
2'(c_1,3(c_2,c_3)) &\to 3(2'(c_1,c_2),c_3)\label{rewr:2'3}\\
2'(c_1,3'(c_2,c_3)) &\to 3'(c_2,2'(c_1,c_3))\label{rewr:2'3'a}\\
2'(c_1,3'(c_2,c_3)) &\to 2'(2'(c_1,c_2),c_3)\label{rewr:2'3'b}\\
2'(c_1,4(c_2,c_2',c_3)) &\to 4(c_2,2'(c_1,c_2'),c_3)\label{rewr:2'4a}\\
2'(c_1,4(c_2,c_2',c_3)) &\to 4(c_2,3(c_1,c_2'),c_3)\label{rewr:2'4b}\\
2'(c_1,4(c_2,c_2',c_3)) &\to 3(3(2'(c_1,c_2),c_2'),c_3)\label{rewr:2'4c}\\
2'(3'(c_1,c_2),c_3) &\to 3'(c_1,2'(c_2,c_3)),\\ %
3(c_1,2(c_2,c_3)) &\to 3(3(c_1,c_2),c_3)\label{rewr:32}\\
3(c_1,2'(c_2,c_3)) &\to 3(3(c_1,c_2),c_3)\label{rewr:32'}\\
3(c_1,3(c_2,c_3)) &\to 3(3(c_1,c_2),c_3)\label{rewr:33}\\
3(c_1,3'(c_2,c_3)) &\to 3(3(c_1,c_2),c_3)\label{rewr:33'}\\
3(3'(c_1,c_2),c_3) &\to 3'(c_1,3(c_2,c_3))\label{rewr:33'*}
\end{align}
}
\end{minipage}
\begin{minipage}[t]{75mm}
{\footnotesize
\begin{align}
3(c_1,4(c_2,c_2',c_3)) &\to 3(3(3(c_1,c_2),c_2'),c_3)\label{rewr:34}\\
3'(2(c_1,c_2),c_3) &\to 3'(c_1,3'(c_2,c_3))\label{rewr:3'2}\\
3'(2'(c_1,c_2),c_3) &\to 3'(c_1,3'(c_2,c_3))\label{rewr:3'2'}\\
3'(3(c_1,c_2),c_3) &\to 3'(c_1,3'(c_2,c_3))\label{rewr:3'3}\\
3'(3'(c_1,c_2),c_3) &\to 3'(c_1,3'(c_2,c_3))\label{rewr:3'3'}\\
3'(4(c_1,c_1',c_2),c_3) &\to 3'(c_1,3'(c_1',3'(c_2,c_3)))\label{rewr:3'4}\\
4(c_1,c_1',2(c_2,c_3)) &\to 4(2(c_1,c_2),c_1',c_3)\label{rewr:42a}\\
4(c_1,c_1',2(c_2,c_3)) &\to 3(4(c_1,c_1',c_2),c_3)\label{rewr:42b}\\
4(c_1,c_1',2(c_2,c_3)) &\to 4(3(c_1,c_2),c_1',c_3)\label{rewr:42c}\\
4(c_1,c_1',2'(c_2,c_3)) &\to 3(4(c_1,c_1',c_2),c_3)\label{rewr:42'a}\\
4(c_1,c_1',2'(c_2,c_3)) &\to 4(c_1,2(c_1',c_2),c_3)\label{rewr:42'b}\\
4(c_1,c_1',2'(c_2,c_3)) &\to 4(c_1,3(c_1',c_2),c_3)\label{rewr:42'c}\\
4(c_1,c_1',3(c_2,c_3)) &\to 3(4(c_1,c_1',c_2),c_3)\label{rewr:43}\\
4(c_1,c_1',3'(c_2,c_3)) &\to 3'(c_2,4(c_1,c_1',c_3))\label{rewr:43'a}\\
4(c_1,c_1',3'(c_2,c_3)) &\to 4(2(c_1,c_2),c_1',c_3)\label{rewr:43'b}\\
4(c_1,c_1',3'(c_2,c_3)) &\to 4(c_1,2'(c_1',c_2),c_3)\label{rewr:43'c}\\
4(c_1,c_1',3'(c_2,c_3)) &\to 3(4(c_1,c_1',c_2),c_3)\label{rewr:43'd}\\
4(c_1,c_1',3'(c_2,c_3)) &\to 4(3(c_1,c_2),c_1',c_3)\label{rewr:43'e}\\
4(c_1,c_1',3'(c_2,c_3)) &\to 4(c_1,3(c_1',c_2),c_3)\label{rewr:43'f}\\
4(3'(c_1,c_2),c_1',c_3) &\to 3'(c_1,4(c_2,c_1',c_3))\label{rewr:43'*a}\\
4(c_1,3'(c_1',c_2),c_3) &\to 3'(c_1',4(c_1,c_2,c_3))\label{rewr:43'*b}\\
4(c_1,c_1',4(c_2,c_2',c_3)) &\to 4(2(c_1,c_2),2'(c_1',c_2'),c_3)\label{rewr:44a}\\
4(c_1,c_1',4(c_2,c_2',c_3)) &\to 4(3(c_1,c_2),2'(c_1',c_2'),c_3)\label{rewr:44b}\\
4(c_1,c_1',4(c_2,c_2',c_3)) &\to 4(2(c_1,c_2),3(c_1',c_2'),c_3)\label{rewr:44c}\\
4(c_1,c_1',4(c_2,c_2',c_3)) &\to 4(3(c_1,c_2),3(c_1',c_2'),c_3)\label{rewr:44d}\\
4(c_1,c_1',4(c_2,c_2',c_3)) &\to 3(3(4(c_1,c_1',c_2'),c_2),c_3)\label{rewr:44e}\\
4(c_1,c_1',4(c_2,c_2',c_3)) &\to 3(3(4(c_1,c_1',c_2),c_2'),c_3)\label{rewr:44f}
\end{align}
}
\end{minipage}
\caption{The term rewrite system $Q$ for composition trees}\label{tab:comp-trs}
\end{table}

\clearpage

Termination of $Q$ follows from a lexicographic combination
of an interpretation $\rho$ that decreases under rotation,
and an interpretation $\sigma$ that decreases under bubbling.

\begin{lemma}\label{lem:r-term}
$Q$ terminates.
\end{lemma}

\begin{proof}
Let the two interpretations $\rho$ and $\sigma$ on natural numbers be defined by
\begin{align*}
\rho(1) &= 2,\\
\rho(2(c_1,c_2)) &= \rho(2'(c_1,c_2)) = \rho(3(c_1,c_2)) = \rho(c_1)+2\rho(c_2),\\
\rho(3'(c_1,c_2)) &= 2\rho(c_1)+\rho(c_2),\\
\rho(4(c_1,c_2,c_3)) &= \rho(c_1)+\rho(c_2)+2\rho(c_3),\\[6pt]
\sigma(1) &= 2,\\
\sigma(2(c_1,c_2)) &= \sigma(2'(c_1,c_2)) = \sigma(3(c_1,c_2)) = \sigma(c_1)\cdot \sigma(c_2),\\
\sigma(3'(c_1,c_2)) &= \sigma(c_1)\cdot \sigma(c_2) + 1,\\
\sigma(4(c_1,c_2,c_3)) &= \sigma(c_1)\cdot \sigma(c_2)\cdot \sigma(c_3) \enspace .
\end{align*}
The order $>$ on terms defined by $s > t$ if $\rho(s) > \rho(t)$ or
$\rho(s) = \rho(t)$ and $\sigma(s) > \sigma(t)$ is a reduction order.
With this, the rules $\ell\to r$
in \ref{rewr:23'*}, \ref{rewr:33'*}, \ref{rewr:43'*a}, and \ref{rewr:43'*b} 
satisfy $\rho(\ell) = \rho(r)$ and  $\sigma(\ell) > \sigma(r)$.
For instance, Rule~\ref{rewr:43'*b} satisfies
$\rho(\ell) = \rho(r) = \rho(c_1)+\rho(c_1')+\rho(c_2)+2\rho(c_3)$ and
$\sigma(\ell) = (\sigma(c_1)\sigma(c_2) + 1)\sigma(c_1')\sigma(c_3) >
\sigma(c_1)\sigma(c_2)\sigma(c_1')\sigma(c_3) + 1 = \sigma(r)$.
All other rules $\ell\to r$ in $Q$ satisfy $\rho(\ell) > \rho(r)$.
For instance, Rule~\ref{rewr:44e} satisfies
$\rho(\ell) = \rho(c_1)+\rho(c_1')+2\rho(c_2)+2\rho(c_2')+4\rho(c_3) >
\rho(c_1)+\rho(c_1')+2\rho(c_2)+2\rho(c_2')+2\rho(c_3)
= \rho(r)$.
So $Q$ is ordered by the reduction order $>$, and so $Q$ terminates.
\end{proof}

The following lemma takes care of the semantic coverage property:
\begin{lemma}\label{lem:one-step}
For every composition tree $c$ that admits a $Q$ rewrite step,
and for every $(s,t)\in\sem{c}$ there
is a composition tree $c'$ such that both $c\to_Q c'$ and
$(s,t)\in\sem{c'}$.
\end{lemma}

\begin{proof}
The proof is done by a case analysis over all left hand sides of $Q$.
We show only one particularly complex case; the other cases work similarly.

Let $c = 4(c_1,c_1',4(c_2,c_2',c_3))$. By definition of $\sem{\cdot}$,
we get $s = \hat{s}wu$, $t = \hat{t}zv$,
$(\hat{s},\hat{t}x)\in \sem{c_1}$, $(u,yv)\in \sem{c_1'}$,
$(xwy,z)\in \sem{4(c_2,c_2',c_3)}$ for some $\hat{t},x,y,v\neq \emptystring$
and some $\hat{s}$.
Again, we get $xwy = s'w'u'$, $z = t'z'v'$,
$(s',t'x')\in \sem{c_2}$, $(u',y'v')\in \sem{c_2'}$,
$(x'w'y',z')\in \sem{c_3}$ for some $t',x',y',v'\neq \emptystring$
and some $s'$.
We distinguish cases according to the relative lengths:
\begin{enumerate}
\item $|x| < |s'|$, $|y| < |u'|$.
Then $(\hat{s}s'',\hat{t}t'x')\in \sem{2(c_1,c_2)}$
where $s''\neq \emptystring$ is defined by $s' = xs''$.
Next, $(u''u,y'v'v)\in \sem{2'(c_1',c_2')}$
where $u''\neq \emptystring$ is defined by $u' = u''y$.
Finally, $(s,t) = (\hat{s}s''w'u''u,\hat{t}t'z'v'v)\in
\sem{4(2(c_1,c_2),2'(c_1',c_2'),c_3)}$,
and we choose $c\to c' = 4(2(c_1,c_2),2'(c_1',c_2'),c_3)$
by Rule~\ref{rewr:44a}.
\item $|s'| \leq |x| < |s'w'|$, $|y| < |u'|$.
Then $(\hat{s},\hat{t}t'x'')\in \sem{3(c_1,c_2)}$
where $x''$ is defined by $x = s'x''$.
Next, $(u''u,y'v'v)\in \sem{2'(c_1',c_2')}$
where $u''\neq \emptystring$ is defined by $u' = u''y$.
Finally, $(s,t) = (\hat{s}w''u''u,\hat{t}t'z'v'v)\in
\sem{4(3(c_1,c_2),2'(c_1',c_2'),c_3)}$,
where $w''$ is defined by $w' = x''w''$,
and we choose $c\to c' = 4(3(c_1,c_2),2'(c_1',c_2'),c_3)$
by Rule~\ref{rewr:44b}.
\item $|x| < |s'|$, $|u'| \leq |y| < |w'u'|$.
This case is symmetric to Case~2. We use Rule~\ref{rewr:44c}.
\item $|s'| \leq |x|$, $|u'| \leq |y|$.
Then $(\hat{s},\hat{t}t'x'')\in \sem{3(c_1,c_2)}$
where $x''$ is defined by $x = s'x''$.
Next, $(u,y''v'v)\in \sem{3(c_1',c_2')}$
where $y''$ is defined by $y = y''u'$.
Finally, $(s,t) = (\hat{s}w''u,\hat{t}t'z'v'v)\in
\sem{4(3(c_1,c_2),3(c_1',c_2'),c_3)}$,
where $w''$ is defined by $w' = x''w''y''$,
and we choose $c\to c' = 4(3(c_1,c_2),3(c_1',c_2'),c_3)$
by Rule~\ref{rewr:44d}.
\item $|s'w'| \leq |x|$, $|y| < |u'|$.
Then $(\hat{s}u''u,\hat{t}y'v')\in \sem{4(c_1,c_2,c_2')}$
where $x''$ is defined by $x = s'w'x''$,
and $u''$ is defined by $u' = x''u''$.
Next, $(\hat{s}u''u,\hat{t}t'x'w'y'v'v)\in \sem{3(4(c_1,c_2,c_2'),c_1')}$.
Finally, $(s,t) = (\hat{s}u''u,\hat{t}t'z'v'v)\in
\sem{3(3(4(c_1,c_2,c_2'),c_1'),c_3)}$,
by Rule~\ref{rewr:44e}.
\item $|x| < |s'|$, $|w'u'| \leq |y|$.
This case is symmetric to Case~5. We use Rule~\ref{rewr:44f}.
\qed
\end{enumerate}
\renewcommand\qed{}
\end{proof}

Now we are ready to prove Theorem~\ref{thm:oc}.
\begin{proof}[Proof of Theorem~\ref{thm:oc}]
For ``$\supseteq$'', we observe that $\CT \supseteq \CT_N$ and hence
$\OC = \sem{\CT} \supseteq \sem{\CT_N} = \OC_N$ by Lemmata~\ref{lem:oc-ct}
and~\ref{lem:oc-nf-neu}.
For ``$\subseteq$'',
we prove that $c\in \OC$ and $(s,t)\in \sem{c}$ implies
$(s,t)\in \OC_N$. We do so by induction on $c$,
ordered by $>$. If $c$ admits a $Q$ rewrite step then
by Lemma~\ref{lem:one-step} there is a composition tree $c'$ such that
both $c\to_Q c'$ and $(s,t)\in\sem{c'}$. Because $c > c'$, the claim follows
by inductive hypothesis for $c'$. Now suppose that $c$ is in $Q$-normal form.
Then $c\in\CT_N$, and so $\sem{c}\subseteq\OC_N$ by Lemma~\ref{lem:oc-nf-neu}.
\end{proof}

From Theorem~\ref{thm:oc}, we immediately get:
\begin{corollary}
$\rhs(\OC) = \rhs(\OC')$.
\end{corollary}

Because $\OC'$ is left-recursive,
we can derive a recursive characterization of the set of right hand sides
of overlap closures:

\begin{corollary}\label{cta:cor:main} (This is Corollary~\ref{ct:cor:main})
$\rhs(\OC)$ is the least set $S$ such that
\begin{enumerate}
\item \label{cta:cor:main:1} $\rhs(R)\subseteq S$,
\item \label{cta:cor:main:2} if $tx\in S$ and $(xu,v)\in R$
for some $t,x,u\neq \emptystring$ then $tv\in S$;
\item \label{cta:cor:main:3} if $xt\in S$ and $(ux,v)\in R$
for some $t,x,u\neq \emptystring$ then $vt\in S$;
\item \label{cta:cor:main:4} if $tut'\in S$ and $(u,v)\in R$
then $tvt'\in S$;
\item \label{cta:cor:main:5} if $tx\in S$ and $yv\in S$ and $(xwy,z)\in R$
for some $t,x,y,v\neq \emptystring$ then $tzv\in S$.
\end{enumerate}
\end{corollary}

\end{document}